\documentclass[11pt,a4paper,reqno]{amsart}
\usepackage{calc,graphicx,amsfonts,amsthm,amscd,epsfig,psfrag,amsmath,amssymb,enumerate,dsfont}

\usepackage[showlabels,sections,floats,textmath,displaymath]{}
\usepackage[colorlinks=true, pdfstartview=FitV, linkcolor=blue, citecolor=blue, urlcolor=blue,pagebackref=false]{hyperref}
\numberwithin{equation}{section}
\DeclareMathSymbol{\leqslant}{\mathalpha}{AMSa}{"36} 
\DeclareMathSymbol{\geqslant}{\mathalpha}{AMSa}{"3E} 
\DeclareMathSymbol{\eset}{\mathalpha}{AMSb}{"3F}     
\renewcommand{\leq}{\;\leqslant\;}                   
\renewcommand{\geq}{\;\geqslant\;}                   

\def\1{\ifmmode {1\hskip -3pt \rm{I}} \else {\hbox {$1\hskip -3pt \rm{I}$}}\fi}

\newtheorem{Theorem}{Theorem}[section]
\newtheorem{Lemma}[Theorem]{Lemma}
\newtheorem{Proposition}[Theorem]{Proposition}
\newtheorem{Corollary}[Theorem]{Corollary}
\newtheorem{remark}{Remark}[section]

\newtheorem{claim}[Theorem]{Claim}
\newtheorem{definition}[Theorem]{Definition}

\newcommand{\cA}{\ensuremath{\mathcal A}}
\newcommand{\cB}{\ensuremath{\mathcal B}}
\newcommand{\cC}{\ensuremath{\mathcal C}}
\newcommand{\cD}{\ensuremath{\mathcal D}}
\newcommand{\cE}{\ensuremath{\mathcal E}}
\newcommand{\cF}{\ensuremath{\mathcal F}}
\newcommand{\cG}{\ensuremath{\mathcal G}}

\newcommand{\cL}{\ensuremath{\mathcal L}}

\newcommand{\cN}{\ensuremath{\mathcal N}}

\newcommand{\cT}{\ensuremath{\mathcal T}}

\newcommand{\cZ}{\ensuremath{\mathcal Z}}

\newcommand{\bbE}{{\ensuremath{\mathbb E}} }

\newcommand{\bbI}{{\ensuremath{\mathbb I}} }

\newcommand{\bbN}{{\ensuremath{\mathbb N}} }

\newcommand{\bbP}{{\ensuremath{\mathbb P}} }

\newcommand{\bbR}{{\ensuremath{\mathbb R}} }

\newcommand{\bbZ}{{\ensuremath{\mathbb Z}} }

\newcommand{\var}{\operatorname{Var}}

\let\a=\alpha \let\b=\beta   \let\d=\delta  \let\e=\varepsilon
 \let\g=\gamma \let\h=\eta      \let\l=\lambda
   \let\n=\nu   \let\o=\omega      
  \let\s=\sigma \let\t=\tau   
  \let\z=\zeta
     \let\L=\Lambda 
\let\O=\Omega      

\def\\{\hfill\break}

\def\thsp{\thinspace}

\def\tthsp{\kern .083333 em}

\def\?{\mskip -10mu}
\def\indbox#1{\hbox to \parindent{\hfil\ #1\hfil} }

\newfam\msafam
\newfam\msbfam
\newfam\eufmfam
\def\hexnumber#1{%
  \ifcase#1 0\or 1\or 2\or 3\or 4\or 5\or 6\or 7\or 8\or
  9\or A\or B\or C\or D\or E\or F\fi}
\font\tenmsa=msam10 \font\sevenmsa=msam7 \font\fivemsa=msam5
\textfont\msafam=\tenmsa \scriptfont\msafam=\sevenmsa
\scriptscriptfont\msafam=\fivemsa
\edef\msafamhexnumber{\hexnumber\msafam}%
\mathchardef\restriction"1\msafamhexnumber16 \mathchardef\ssim"0218
\mathchardef\square"0\msafamhexnumber03
\mathchardef\eqd"3\msafamhexnumber2C
\def\QED{\ifhmode\unskip\nobreak\fi\quad
  \ifmmode\square\else$\square$\fi}
\font\tenmsb=msbm10 \font\sevenmsb=msbm7 \font\fivemsb=msbm5
\textfont\msbfam=\tenmsb \scriptfont\msbfam=\sevenmsb
\scriptscriptfont\msbfam=\fivemsb

\font\teneufm=eufm10 \font\seveneufm=eufm7 \font\fiveeufm=eufm5
\textfont\eufmfam=\teneufm \scriptfont\eufmfam=\seveneufm
\scriptscriptfont\eufmfam=\fiveeufm

\def\({\left(}
\def\){\right)}
\let\neper=e
\let\ii=i

\def\ie{\hbox{\it i.e.\ }}
\let\id=\identity

\let\sset=\subset

\def\nep#1{ \neper^{#1}}

\def\tc{\thsp | \thsp}
\let\<=\langle
\let\>=\rangle

\def\Var{ \mathop{\rm Var}\nolimits }

\def\gap{\mathop{\rm gap}\nolimits}

\def\inte#1{\lfloor #1 \rfloor}

\outer\def\nproclaim#1 [#2]#3. #4\par{\medbreak \noindent
   \talato(#2){\bf #1 \Thm[#2]#3.\enspace }%
   {\sl #4\par }\ifdim \lastskip <\medskipamount
   \removelastskip \penalty 55\medskip \fi}
\def\thmm[#1]{#1}
\def\teo[#1]{#1}
\def\sttilde#1{%
\dimen2=\fontdimen5\textfont0 \setbox0=\hbox{$\mathchar"7E$}
\setbox1=\hbox{$\scriptstyle #1$} \dimen0=\wd0 \dimen1=\wd1
\advance\dimen1 by -\dimen0 \divide\dimen1 by 2
\vbox{\offinterlineskip%
   \moveright\dimen1 \box0 \kern - \dimen2\box1}
}
\begin{document}
\title[Aging through hierarchical coalescence in the East model]{Aging through hierarchical coalescence \\in the East model}

\author{A. Faggionato}
\address{Alessandra Faggionato. Dip. Matematica ``G. Castelnuovo", Univ. ``La
  Sapienza''. P.le Aldo Moro  2, 00185  Roma, Italy. e--mail:
  faggiona@mat.uniroma1.it}

 \author[F. Martinelli]{F. Martinelli}
 \address{F. Martinelli. Dip. Matematica, Univ. Roma Tre, Largo S.L.Murialdo 00146, Roma, Italy. e--mail:
martin@mat.uniroma3.it
 }

 \author[C. Roberto]{C. Roberto}
 \address{Cyril Roberto. L.A.M.A., Univ. Marne-la-Vall\'ee, 5 bd Descartes 77454 Marne-la-Vall\'ee, France. e--mail:
cyril.roberto@univ-mlv.fr
 }

\author[C. Toninelli]{C. Toninelli}
\address{Cristina Toninelli. L.P.M.A. and
  CNRS-UMR 7599, Univ. Paris VI-VII 4, Pl. Jussieu 75252
  Paris, France. e--mail: cristina.toninelli@upmc.fr}

\thanks{Work supported by the European Research Council through the ``Advanced
Grant'' PTRELSS 228032.}

\begin{abstract}
We rigorously analyze the low temperature non-equilibrium dynamics of the East model,
a special example of a one dimensional oriented kinetically
constrained particle model, when the initial distribution is different from the reversible one and for times much smaller than the global relaxation time.
This setting has been intensively studied in the physics literature to
analyze the slow dynamics which follows a sudden quench from the liquid to the glass phase.
In the limit of  zero temperature (\ie a vanishing density of
vacancies) and for initial distributions such that the vacancies form
a renewal process we prove that the density of vacancies, the persistence function and the two-time autocorrelation function behave as staircase functions
with several plateaux. Furthermore the two-time autocorrelation function   displays an aging behavior. We
also provide a sharp description of the statistics of the domain
length as a function of time, a domain being the interval between two consecutive vacancies.
When the initial renewal process has finite mean our results
confirm (and generalize) previous findings of the physicists for the
restricted case of a product Bernoulli measure. However we show that a
different behavior appears when the initial domain distribution is in
the attraction domain of a $\a$-stable law. All the above results actually
follow from a more general result which says that the low
temperature dynamics of the East model is very well described by that
of a certain \emph{hierarchical coalescence process}, a probabilistic
object
which can be viewed as a hierarchical sequence of suitably linked
coalescence processes and whose asymptotic behavior has been recently
studied
in \cite{FMRT}.  \\

\noindent \textit{Mathematics Subject Classification: 60K35, 82C20.
}
\\

\noindent
\textit{Keywords: kinetically constrained models, non-equilibrium dynamics, coalescence, metastability, aging, interacting particle systems.}
 \end{abstract}

\maketitle

\thispagestyle{empty}

\section{Introduction}
Facilitated or kinetically constrained spin (particle) models (KCSM) are
interacting particle systems which have been introduced in the physics
literature \cite{FA1,FA2,JACKLE} to model liquid/glass transition and more
generally ``glassy dynamics'' (see  e.g. \cite{Ritort,GarrahanSollichToninelli}). A configuration is given by
assigning to each vertex $x$ of a (finite or infinite) connected graph
$\cG$ its occupation variable
$\eta(x)\in\{0,1\}$ which corresponds to an empty or filled site,
respectively.  The evolution is given by a Markovian stochastic dynamics
of Glauber type.  Each site with rate one refreshes its occupation variable to
a filled or to an empty state with probability $1-q$ or $q$
respectively provided that the current configuration around it satisfies an
a priori specified constraint.  For each site $x$ the corresponding constraint does not
involve $\h(x)$, thus detailed balance w.r.t. the Bernoulli($1-q$) product
measure $\pi$ can be easily verified and the latter is an invariant
reversible measure for the process.

One of the most studied KCSM is the East model \cite{JACKLE}.
It  is  a one-dimensional  model ($\cG=\bbZ$ or $\cG=\bbZ_+=\{0,1,\dots\}$) and particle
creation/annihilation at a given site $x$ can occur only if its right neighbor $x+1$ is empty. The model is ergodic for any
$q\neq 0,1$ with a positive spectral gap \cite{Aldous,CMRT} and it relaxes to
the equilibrium reversible measure exponentially fast even when started
from e.g. any non-trivial product measure \cite{CMST}. However, as $q\downarrow
0$, the relaxation time $T_{\rm relax}(q)$ diverges very fast,  
$T_{\rm relax}\sim\left(\exp(\l\log(1/q)^2)\right)$ with a sharp constant $\l$ (see \cite{CMRT}).
A key issue, both from the mathematical and the physical point of view, is
therefore that of describing accurately the evolution at $q\ll 1$ when the initial distribution is different from the reversible one
and for time scales which are large but still much smaller
than $T_{\rm relax}(q)$ when the exponential relaxation to the
reversible measure takes over. 

 An initial distribution  which is often considered in the physics
 literature is the Bernoulli  distribution at a density $1/2$ \cite{SE1,SE2}. We refer the interested reader to \cite{Ritort,Leonard,Crisanti,GarrahanNewman,CorberiCugliandolo} for the relevance of this setting in connection with the study of the liquid\/glass transition as well as for details for KCMS different from East model. 

Let us give a rough picture of the non-equilibrium dynamics of the East model as $q\downarrow 0$.
Since the equilibrium vacancy
density is very small, most of the non-equilibrium
evolution will try to remove the excess of vacancies present in the initial
distribution and will thus be dominated by the coalescence of domains corresponding to the intervals separating two consecutive vacancies.
Of course this process must necessarily occur in a kind
of cooperative way because, in order to
remove a vacancy, other vacancies must be created nearby (to its
right). Since the creation of vacancies requires the overcoming
of an \emph{energy barrier}, in a first
approximation the
non-equilibrium dynamics  of the East model for $q\ll 1$ is driven by
a non-trivial energy landscape.

In order to better explain the structure of this landscape suppose that we start from a configuration
with only two vacancies located at the sites $a$ and $a+\ell$,  with $\ell\in
[2^{n-1}+1,\dots, 2^n]$. In this case a nice combinatorial argument (see
\cite{CDG} and also \cite{SE2}) shows that, in order to remove the
vacancy at $a$ within time $t$, there must exists $s\le t$ such
that the number of vacancies inside the interval $(a,a+\ell)$ at time
$s$ is at least $n$. It is rather easy to show that at any given time
$s$ the probability of observing $n$ vacancies in $(a,a+\ell)$ is
$O(q^n)$ so that, in order to have a non negligible probability of
observing the disappearance of the vacancy at $a$, we need to wait an
\emph{activation time} $t_n=O(1/q^n)$. In a more physical language the energy barrier which the system
must overcome is $O(\log_2 \ell)$. As it is the case in many
metastable phenomena,
once the system decides to overcome the barrier and kill the vacancy,
it does it in a time scale much smaller than the activation time. In
our case this scale is $t_{n-1}=1/q^{n-1}$.

The above argument indicates the following heuristic picture.
\begin{enumerate}[(i)]
\item  A \emph{hierarchical structure} of
the activation times $t_n=1/q^n$ (and of the energy landscape) well separated
one from the other for $q\ll 1$.
\item  A kind of metastable behavior of the dynamics which removes
vacancies in a hierarchical fashion.
\item  Since the characteristic time scales $t_n$ are well separated
  one from the other, the evolution should show \emph{active}  and
  \emph{stalling} periods. During the $n^{th}$-active period, identified with e.g. the interval  $[t_n^{1-\epsilon},t_n^{1+\epsilon}]$, $\epsilon \ll
  1$, only the vacancies with another vacancy to their right at
  distance less than $2^n$ can be removed. At the end of an active
  period no vacancies with distance less than  $2^n+1$ are present
  anymore as
  well as
  no extra (\ie not present at time $t=0$) vacancies. During the
  $n^{th}$-stalling period $[t_n^{1+\epsilon},t_{n+1}^{1-\epsilon}]$
  nothing interesting happens in the sense that none of the vacancies present at the beginning of the period are
  destroyed and no new vacancies are created at the end of the
  period.
\end{enumerate}
Clear the above scenario, and particularly the presence of active and
stalling periods, implies that physical
  quantities like the persistence function or the density of vacancies
  should behave as a staircase function with several  \emph{plateaux}
  and that
  \emph{aging} should occur for two-time quantities as the two time-autocorrelation.

Such a general picture was somehow suggested in two interesting physics papers
\cite{SE1,SE2} and some of the conclusions (properties (iv) above) were indeed observed in numerical
simulations \cite{SE1,Leonard}. In \cite{SE1} the true East dynamics
was replaced with that of a certain \emph{hierarchical coalescence model}
mimicking  the features (i)--(iii) described above. In turn, under the assumption that  the interval between two consecutive vacancies (domain) in the $n$-th stalling period rescaled by $2^n$ has a well defined limiting distribution
as $n\to\infty$, the form of this limiting distribution when the initial distribution is a Bernoulli product measure has been computed for the coalescence model.

Partly motivated by the above discussion and partly by other
coalescence models in statistical physics with a mean field structure (see
e.g. \cite{D0,D1,D2,D3}), the present authors introduced in
\cite{FMRT} a large class of hierarchical coalescence models and: (1) proved the existence of a scaling limit under very general
assumptions, (2) proved the universality of the scaling limit depending only on general features of the initial distribution
and not on the details of the model. We refer the reader to Section
\ref{section:hcp} for more details and to \cite{FMRT} for a much more
general setting.

In this paper, besides providing a mathematical derivation of the
above mentioned heuristic picture,  we rigorously establish aging and plateau behavior (Theorem \ref{plateau}).
Furthermore (Theorem \ref{asymptotics}) we prove
a scaling limit for: i) the
inter-vacancy distance and ii) the position of the first vacancy for
the model on the positive half line. In particular we prove that this scaling limit  is universal if the initial renewal process has finite
mean. If instead the initial distribution is the domain of attraction
of an $\a$-stable law, $\a\in (0,1)$, the scaling limit is different
and falls
in another universality class depending on $\a$. In order to establish the above results we actually prove a result which is more fundamental and of
independent interest. Namely we show that with probability tending to
one as $q\downarrow 0$, the non equilibrium dynamics of the East model starting from a
renewal process is well approximated (in variation distance) by a
suitable hierarchical coalescence process with rates depending on
suitable large deviation probabilities of the East model (Theorem \ref{bacio}).

\tableofcontents

\section{The East process: definition and main results}
\noindent In what follows we will use the notation $\bbN:=1,2,\dots$
and $\bbZ_+:=0,1,2,\dots$. We will focus on the East process on $\bbZ_+$ and explain in Section \ref{extensions} how the result can be extended to the process defined on $\bbZ$.
The East process  on $\bbZ_+$ with parameter
$q\in[0,1]$ is an interacting particle systems with a Glauber type dynamics
on the configuration space $\Omega:=\{0,1\}^{\bbZ_+}$, reversible with respect to  the product
probability measure $\pi:=\prod_{x\in\bbZ_+}\pi_x$,  $\pi_x$ being the
Bernoulli$(1-q)$ measure. Since we are interested in the small $q$
regime throughout the following   we will assume $q\le 1/2$.
\begin{remark}\label{temp}
Sometimes in the physical literature the parameter
$q$ is written as $q=\frac{\nep{-\b}}{1+\nep{-\b}}$ where $\b$ is the inverse temperature
so that the limit $q\downarrow 0$ corresponds to the zero temperature limit.
\end{remark}
Elements of $\O$ will usually be denoted by the Greek
letters $\sigma,\eta,\dots$ and $\s(x)$ will denote the occupancy
variable at the site $x$. The restriction of a configuration $\s$ to a
subset $\L$ of $\bbZ_+$ will be denoted by $\s_\L$.
The set of empty sites (or zeros in the sequel) of a configuration
$\s$ will be denoted by $\cZ(\s)$ and they will often be referred to as $x_0<x_1<\dots$
without the specification of the configuration if clear from the context.

The East  process can be informally described as follows. Each vertex
$x$ waits an independent mean one exponential time and then, provided that the current configuration $\sigma$ satisfies the constraint $\sigma(x+1)=0$, the value of $\sigma(x)$ is refreshed  and set equal to $1$ with probability $1-q$ and to $0$ with probability $q$.
Formally (see  \cite{Liggett1})  the process is uniquely specified by the action of its infinitesimal Markov generator $\cL$
on local (\ie depending on finitely many variables) functions $f:
\O\mapsto \bbR$ that is given by
\begin{align}
\label{thegenerator}
\cL f(\sigma)
& =
\sum_{x\in \bbZ_+}c_{x}(\sigma)\left[\pi_x(f)-f(\sigma)\right] \\
& =
\sum_{x\in \bbZ_+}c_{x} (\sigma)\left[(1-\sigma(x))(1-q)+\sigma(x)q\right]\left( f(\sigma^x) - f(\sigma) \right) \nonumber
\end{align}
where $c_x(\sigma):=1-\sigma(x+1)$ encodes the constraint,
$\pi_x(f)$ denotes the conditional mean $\pi(f\tc \{\s(y)\}_{y\neq
  x})$ and $\sigma^x$ is obtained from $\sigma$ by flipping its value at $x$, {\it i.e.}
$$
\sigma^x(y)=\left\{
\begin{array} {ll}
\sigma(y) & \mbox{if } y \neq x \\
1-\sigma(x) & \mbox{if } y=x
\end{array}
\right. .
$$
When the initial distribution at time $t=0$ is $Q$ the law and
expectation of the
process on the Skohorod space $D([0,\infty) , \O)$ will be denoted by
$\bbP_Q$ and $\bbE_Q$ respectively. If $Q=\d_{\s}$ we write simply $\bbP_\s$.  In the sequel we will often write $x_k(t)$ for
the $k^{th}$-zero for the process $\s_t$ at time $t$ if no confusion arises.

\begin{definition}
Given two probability measures $\mu$ on $\bbN:=[1,2,\dots)$ and
$\nu$ on $\bbZ_+$ we will write $Q=\text{Ren}(\nu,\mu)$ if, under $Q$,
the first zero $x_0$ has law $\nu$ and it is independent of the random variables
$\{x_k-x_{k-1}\}_{k=1}^\infty$ which, in turn, form
a sequence of i.i.d random variables with
common law $\mu$. If $\nu=\d_0$ then we will write
$Q={\rm Ren}(\mu\tc 0)$.
\end{definition}
\begin{remark}
In most of the present paper the initial distribution $Q$ will always be
assumed to be of the above form. For further generalizations we refer to
Section \ref{extensions}.
\end{remark}
The East process can  also be defined on finite intervals $\L:=[a,b]\sset
\bbZ_+$ provided that a suitable \emph{zero} boundary condition is specified at the
site $b+1$.
More precisely one defines the finite volume generator
\begin{gather}\cL_{\L} f(\s)= \sum_{x\in [a,b-1]}c_{x}(\s)\left[\mu_x(f)-f(\s)\right]+ \left[\mu_b(f)-f(\s)\right]
    \equiv \sum_{x\in \L}c_{x}^{\L}(\s)\left[\mu_x(f)-f(\o)\right]\, ,\nonumber\\
   \text{where}\quad c_x^{\L}(\s)=\begin{cases}
    1-\s(x+1) &\text{for $x\in[a,b-1]$}\\1 &\text{if $x=b$}
   \end{cases}
\label{bersani}
\end{gather}
In particular there is no constraint at site $b$, a fact that pictorially we can
interpret by saying that there is a \emph{frozen zero} at site
$b+1$. This frozen zero is the above mentioned boundary condition.
In this case the process is nothing but a continuous time
Markov chain reversible w.r.t. the product measure
$\pi_{\L}:=\prod_{x\in [a,b]}\pi_x$ and, due to the ``East''
  character of the constraint, for any initial condition $\h$ its
  evolution coincides with that of the East process in $\bbZ_+$
  (restricted to $\L$)
  starting from the configuration
    \begin{equation}\label{estendiamo}
    \tilde \h (x):=
    \begin{cases} \h(x) & \text{ if } x\in [a,b]\,,\\
    0 & \text{ if }x=b+1\\
     1 & \text{ otherwise}\,.
     \end{cases}
     \end{equation}
We will use the self-explanatory notation $\bbP_Q ^{\L}$ (or
$\bbP_\s^{\L}$) for the law of the process starting from the law $Q$ (from $\s$).

\subsubsection{Additional notation}
In the sequel $\L$ will always denote a finite interval of $\bbZ_+$
with endpoints $0\le a<b<\infty$.

It will also be quite useful to isolate some special
configurations in $\O_\L$.
We
denote by $\s_{0\mathds{1} }$ the configuration in $\O_\L:=\{0,1\}^\L$
such that $\cZ(\s)=\{a\}$ and by $\s_{\mathds{1}}$ the configuration
  with $\cZ(\s)=\emptyset$. In words,
   $\s_{0\mathds{1} }$ is the configuration with a single zero located at the left
   extreme of the interval, while $\s_\mathds{1}$ is the configuration with no zeros.
     We also let, with a slight abuse of notation, $\bbP^{\L}_{0\mathds{1}}:=\bbP ^{\L}_{\sigma_{0\mathds{1}}}$
      and  $\bbP^{\L}_{\mathds{1}}:=\bbP
      ^{\L}_{\sigma_{\mathds{1}}}$.


\subsection{Graphical construction}
\label{graphical}
Here we recall a standard graphical construction which allows to
define on the same probability space the finite volume East process for \emph{all} initial
conditions.
Using a standard percolation argument \cite{Durrett,Liggett2} together with the fact
that the constraints $c_x$ are uniformly bounded and of finite range,
it is not difficult to see that the graphical construction can be extended without problems also to
the infinite volume case.
Given a finite interval $\L\sset \bbZ_+$ we associate to each $x\in \L$  a  Poisson process  of parameter one  and, independently, a family of
independent Bernoulli$(1-q)$  random variables $\{s_{x,k}:k\in
\bbN\}$. The occurrences of the Poisson process associated to $x$ will
be denoted by $\{t_{x,k}:\ k\in \bbN\}$. We assume independence as $x$
varies in $\L$. Notice that with probability one all the
occurrences $\{t_{x,k}\}_{k\in \bbN,\, x\in \bbZ_+}$ are
different. This defines the probability space. The corresponding
probability measure will be denoted by $\bbP$. Given an initial
configuration $\h\in \O$ we construct a Markov process
$(\s_t^{\L,\h})_{t\ge 0}$ on the above probability space satisfying
$\s^{\L,\h}_{t=0}=\h$  according to the
following rules.
At each time $t=t_{x,n}$ the site $x$ queries the state of its own
constraint $c^\L_x$. If the constraint  is satisfied, \ie if
$\s^{\L,\h}_{t-}(x+1)=0$, then $t_{x,n}$ will be  called a \emph{legal ring} and
at time $t$ the configuration resets its value at site $x$  to the value of the corresponding Bernoulli variable
$s_{x,n}$. We stress here that
the rings and coin tosses at $x$ for $s\le t$ have no influence whatsoever on the
evolution of the configuration at the sites which enter in its constraint (here
$x + 1$) and thus they have no influence of whether a ring at $x$ for $s > t$ is legal
or not. It is easy to check that the above construction actually gives a
continuous time Markov chain with generator \eqref{bersani}.

A first immediate consequence is
the following \emph{decoupling} property.
\begin{Lemma}
\label{oriented} Fix  $c<a< b< d$ with
$a,b,c,d\in\bbZ_+\cup\{\infty\}$ and let $\L=[c,d],\ \L'=[a,b],\ V=[b+1,d]$.  Take two events ${\mathcal{A}}$
and ${\mathcal{B}}$, belonging  respectively to the $\s$--algebra
generated by $\{\sigma_s (x) \}_{s\leq t,\, x\in \L'}$ and
$\{\sigma_s(x)\}_{s\leq t,\, x\in V}$. Then, for any $\sigma\in\Omega_\L$,\\
(i)
  $\bbP_{\sigma}^{\L}({\mathcal{B}})=\bbP_{\s_{V} }^{V}({\mathcal{B}})$;\\
(ii) $\bbP_{\sigma  }^{\L}({\mathcal{A}}\cap{\mathcal{B}}\cap\{
\sigma_s^{\L}(b+1)=0 ~\forall s\leq t\})
 =\bbP_{\sigma_{\L'} }^{\L'}({\mathcal{A}})\bbP_{\sigma_{V} }^{V}({\mathcal{B}}\cap\{\sigma_s^{V}(b+1)=0 ~\forall s\leq t\})\,.
$
\end{Lemma}
The last, simple but quite important consequence of the graphical
construction is the following one. Assume
that the zeros of the starting configuration $\s$ are labeled in
increasing order as $x_0,x_1,\dots,x_n$ and define $\t$ as the first time
at which one the $x_i$'s is killed, \ie the occupation variable there flips to
one. Then, up to time $\t$ the East dynamics factorizes over the East
process in each interval $[x_i, x_{i+1})$.
\subsection{Ergodicity}
The finite volume East process is trivially ergodic because of the
frozen zero boundary condition (see \ref{bersani}). The infinite
volume process in $\bbZ_+$ is ergodic in the sense that $0$
is a simple eigenvalue of the generator $\cL$ thought of as a
selfadjoint operator on $L^2(\O,\pi)$ \cite{CMRT}. As far as more
quantitative results are concerned we recall the following (see  \cite{CMRT} for part (i) and \cite{CMST} for part (ii)).
\begin{Theorem}\
  \begin{enumerate}[(i)]
  \item The generator \eqref{thegenerator} has a
    positive spectral gap, denoted by
    $\gap(\cL)$, such that
    \begin{equation*}
      \lim_{q\downarrow 0}\log(\gap(\cL)^{-1})/\left(\log(1/q)\right)^2=(2\log 2)^{-1}.
    \end{equation*}
Moreover, for any interval $\L$, the spectral gap of the finite volume
generator $\cL_\L$ is not smaller than $\gap(\cL)$.
\item Assume that the initial distribution $Q$ is a product
  Bernoulli($\a$) measure, $\a\in (0,1)$. Then there exists $m\in (0,\gap(\cL)]$ and for any local function
  $f$ there exists a constant $C_f$ such that
  \begin{equation*}
    |\bbE_Q(f(\s_t))-\pi(f)|\le C_f \nep{-mt}
  \end{equation*}
 \end{enumerate}
\end{Theorem}
The above results show that relaxation to equilibrium is indeed taking
place at an exponential rate on a time scale $T_{\rm relax}=\gap(\cL)^{-1}$
which however, for small values of $q$, is very large and of the order of
$ \nep{c \log(1/q)^2}$ with $c=(2\log 2)^{-1}$.

\subsection{Main
  results: plateau behavior, aging and scaling limits}

 We are now ready to state our first set of results (Theorem \ref{plateau} and \ref{asymptotics}) which details the non equilibrium behavior of
the East process for small values of $q$ (small temperature) and for time
scales much smaller that $T_{\rm relax}$.
The prove of both theorems is detailed in Section \ref{sectionproofs} and is obtained thanks to the approximation of the East model with a suitable coalescence process. The definition of this coalescence process and the approximation result  (Theorem \ref{bacio}), which is indeed the heart of our paper,  is instead  stated in Section \ref{section:hcp} and proven in Section \ref{sectionbacio}.

\begin{definition}
\label{stalling-active}
Given $\epsilon,\, q\in (0,1)$,   we set
\begin{eqnarray}
& t_0:=1; ~~~~~t_0^-:=0; ~~~~~t_0^+=\left(\frac{1}{q}\right)^{\epsilon}\nonumber\\
& t_n:=\left(\frac{1}{q}\right)^n;~~~~~t_n ^-:=t_n ^{1-\e};~~~~~t_n
^+=t_n ^{1+\e}\,\,\,\,\,\,\forall n\geq 1\,.\label{deftn}
\end{eqnarray}
The time interval $[t_n^-, t_n^+]$ and $[t_n^+,t_{n+1}^-]$ will be
called respectively the {\sl $n^{th}$-active period} and the {\sl $n^{th}$-stalling
period}.
\end{definition}
\begin{Theorem}[Persistence, vacancy density and two-time autocorrelations during stalling periods: plateau and aging]
\label{plateau}
Assume that the initial distribution $Q$ is a renewal measure
$Q={\rm Ren}(\mu\tc 0)$ with $\mu$  such that, for any $k\in \bbN$, $\mu\left([k,\infty)\right) >0$ and either one of
the following holds:
\begin{enumerate}[a)]
\item $\mu$ has finite mean;
\item $\mu$ belongs to the domain of attraction of a $\a$-stable law
  or, more generally, $\mu((x,+\infty))=x^{-\a}L(x)$ where $L(x)$ is a
  slowly varying function at $+\infty$, $\a\in [0,1]$\footnote{A function $L$ is said to
be slowly varying at infinity, if, for all $c>0$, $\lim\limits_{x \to
\infty} L(cx)/L(x)=1$.}.
\end{enumerate}
Then, if $o(1)$ denotes an error term depending only
on $n,m$ and tending to
zero as both tend to infinity,
\begin{enumerate}[(i)]
\item \begin{align}
& \lim_{q\downarrow 0}\, \sup _{t \in [t_{n}^+, t_{n+1}^-]}
\left|\bbP_Q(\sigma_t(0)=0)-\left(\frac{1}{2^{n}+1}\right)^{c_0(1+o(1)) }\right|=0\,,\\
& \lim _{q\downarrow 0}\, \sup _{t \in [t_{n}^+,
t_{n+1}^-]}\left|\bbP_Q(\sigma_s(0)=0~~~\forall s\leq
t)-\left(\frac{1}{2^{n}+1}\right)^{c_0(1+o(1)) }\right|=0\,,
\end{align}
where $c_0=1$ in case (a) and $c_0=\a$ in case (b).
\item Let $t,s:[0,1/2]\to [0,\infty)$ with $t(q)\geq s(q)$ for all $q\in [0,1/2]$. Then
$$\varlimsup_{q\downarrow 0}\,  \bbP_Q(\sigma_{t(q)}(0)=0)\leq \varlimsup _{q\downarrow 0}\,
 \bbP_Q(\sigma_{s(q)}(0)=0).$$
 The same bound holds with $\varliminf_{q\downarrow 0}$ instead
 of $\varlimsup_{q\downarrow 0}$.
\item For $x\in \bbZ_+$ consider the time auto-correlation function
$C_Q(s,t,x):=\text{\rm Cov}_Q(\sigma_t;\sigma_s)$. Then, for any $n,m$,
$$
\lim _{q\downarrow 0}\, \sup _{\stackrel{t \in [t_{n}^+, t_{n+1}^-]}{s \in [t_{m}^+, t_{m+1}^-]}}
\left|C_Q(s,t,x)
  -\rho_x\left(\frac{1}{2^{n}+1}\right)^{c_0(1+o(1))}\left(1- \rho_x \left(\frac{1}{2^{m}+1}\right)^{c_0(1+o(1))}
    \right)\right|=0$$
where $\rho_x=Q(\s(x)=0)$.
\end{enumerate}
\end{Theorem}

  The picture that emerges from points (i) and (ii) is depicted in Figure \ref{fig:plateau}

\begin{figure}[h]
\psfrag{T}{ $\!\!\!\! \frac{\log t}{|\log q|}$}
 \psfrag{0}{\footnotesize $\epsilon$}
 \psfrag{1}{\footnotesize $ 1-\epsilon$}
 \psfrag{2}{\footnotesize $1+\epsilon$}
 \psfrag{3}{\footnotesize$2-\epsilon$}
 \psfrag{4}{\footnotesize$2+\epsilon$}
 \psfrag{5}{\footnotesize$n+\epsilon$}
 \psfrag{6}{}
\psfrag{cn}{$c_n$}
\psfrag{p}{$\!\!\!\!\!\!\!\! \mathbb{P}_{\mathcal{Q}}(\sigma_t(0)=0)$}
 \includegraphics[width=.90\columnwidth]{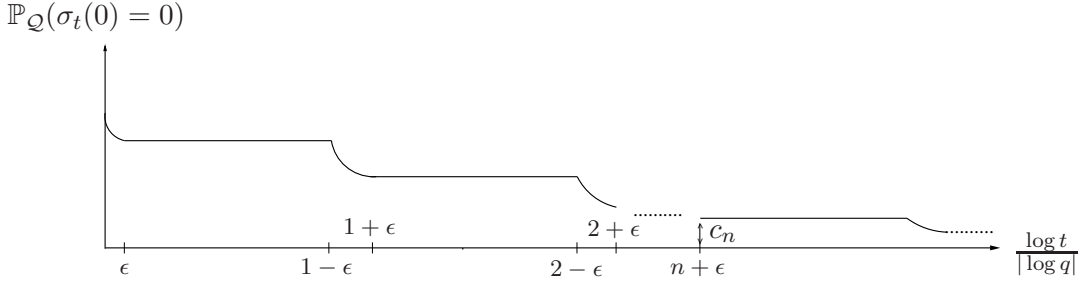}
\caption{Plateau behavior in the limit $q \to 0$, where  we set $c_n:=\left(1/(2^{n}+1)\right)^{c_0(1+o(1))}$ with $c_0$ defined in Theorem \ref{plateau} and $o(1)$ going to zero as $n\to\infty$.}
\label{fig:plateau}
 \end{figure}

\begin{remark}
\label{anyk} \ \\
(1.)  Parts (ii)-(iii) hold also if    $Q=\text{Ren}(\nu,\mu)$ (the proof given in Section \ref{sectionproofs} remains unchanged).
On the other hand part (i)   holds  for $Q=\text{Ren}(\nu,\mu)$   multiplying the asymptotic value by  the factor $Q(\s_0(0)=0)$.
Alternatively, part (i) holds for $Q=\text{Ren}(\nu,\mu)$
if  the site $x=0$ is replaced by
the position $x_k$ of the $k^{th}$-zero at time $t=0$, $k\ge 0$. In fact, because of the renewal property of $Q$ and of the ``East'' feature of the process (see e.g. \eqref{estendiamo}), for any $a\in \bbN$ it holds that
$\bbP_Q(\s_t(a)=0\tc x_k(t=0)=a)= \bbP_{\hat Q}(\s_t(0)=0)$ where
$\hat Q=\text{Ren}(\mu\tc 0)$.    \\
(2.) For small values of $q$ the time auto-correlation function
$C_Q(s,t,x)$ does depend in a non trivial way on $s,t$ and not just
on their difference $t-s$. Hence the word ``aging'' in the title. Of
course, for times much larger than the relaxation time $\gap^{-1}$,
the time auto-correlation will be very close to that of the
equilibrium process which in turn, by reversibility, depends only on $t-s$.
\end{remark}

The next theorem describes the statistics of the interval (domain)
between two consecutive zeros in a stalling period.

In order to state it let, for any $c_0\in (0,1]$, $\tilde X^{(\infty)}_{c_0}\ge 1$ be
a random variable with Laplace transform given  by
\begin{equation}\label{macedonia}
\bbE(\nep{-s\tilde X^{(\infty)}_{c_0}}) =1- \exp\Big \{ - c_0 \int _1^\infty
\frac{e^{-sx}}{x} dx \Big\}= 1-\exp\Big \{ - c_0 \, \text{Ei}(s)
\Big\}\,.
\end{equation}
The corresponding probability density is of the
form $p_{c_0} (x)\bbI_{x \geq 1}$ where $p_{c_0}$ is
the continuous function on $[1, \infty)$ given by
\begin{equation}\label{vonnegut}
 p_{c_0}(x)= \sum _{k=1}^\infty \frac{(-1)^{k+1}c_0^k}{k!}\, \rho_k(x)
 {\id}_{x\geq k }\,,
 \end{equation}
 where $ \rho _1(x)= 1/x$ and
 \begin{equation}\label{kurt}
   \rho_{k+1} (x)= \int _1^\infty d x_1
  \cdots \int_1 ^\infty dx _k \frac{1}{ x-\sum_{i=1}^{k} x_i } \prod
 _{j=1}^{k} \frac{1}{x_j} \,,\qquad  k\geq 1\,.
 \end{equation}
Let also $\tilde Y^{(\infty)}_{c_0}$ be
a non-negative random variable with Laplace transform given  by
\begin{equation}\label{macedonia2}
\bbE(\nep{-s \tilde Y^{(\infty)}_{c_0}}) :=1- \exp\Big \{ - c_0 \int _0^1
\frac{e^{-sx}}{x} dx \Big\}
\end{equation}
\begin{Theorem} [Limiting behavior of the domain length and of the position of the first zero in the stalling
  periods]
\label{asymptotics}
In the same assumptions of Theorem \ref{plateau}, let
$$
\bar X^{(n)}(t):=(x_{1}(t)-x_0(t))/(2^{n-1}+1)\quad ;\quad \bar
Y^{(n)}(t):=x_0(t)/(2^{n-1}+1).
$$
 Then, for any bounded function $f$,
 \begin{align}
&\lim _{n\uparrow \infty}
 \lim_{q\downarrow 0}  \sup_{t \in [t_n^+, t_{n+1}^- ]}
 \Bigl| \bbE_Q\bigl( f( \bar  X^{(n+1)}(t) ) \bigr)- E\bigl(f(\tilde
 X_{c_0}^{(\infty)}\bigr)\Bigr|=0  \label{spa}\\
&\lim _{n\uparrow \infty}\lim_{q\downarrow 0}  \sup_{t \in [t_n^+, t_{n+1}^-]}
 \Bigl| \bbE_Q\bigl( f( \bar  Y^{(n+1)}(t) ) \bigr)- E\bigl(f(\tilde
 Y_{c_0}^{(\infty)}\bigr)\Bigr|=0
\label{spaghetti1}
\end{align}
where again $c_0=1$ if $\mu$ has finite mean and $c_0=\a$ if  $\mu$
belongs to the domain of attraction of a $\a$-stable law.\\
The result \eqref{spa} holds for $f$ satisfying $|f(x)|\leq C (1+|x|)^m$,
$m=1,2,\dots$, if the $(m+\d)^{th}$-moment of $\mu$ and $\nu$ is finite for some
$\d>0$.
\end{Theorem}
\begin{remark}
The above result holds also for $Q={\rm Ren}(\nu,\mu)$, which can be
obtained from $Q={\rm Ren}(\mu\tc 0)$ by a random shifting of law
$\nu$. Trivially, the effect of this random shift disappears in the
scaling limit.
 Moreover the moment condition can be relaxed (see the proof of Proposition \ref{finvol2}).
\end{remark}

\section{Hierarchical Coalescence and the East process}
\label{section:hcp}
  In this section we introduce a  {\sl
  hierarchical coalescence process} (in the sequel HCP) which
 belongs to a much larger class of  processes  whose definition, asymptotic behavior and scaling limits  are stated and analyzed in \cite{FMRT}.
We will then state a result (Theorem \ref{bacio}) which had been conjectured in \cite{SE1} which says that  in the low temperature limit
$q\downarrow 0$ the East process is well described by HCP. This result, together with the knowledge of the asymptotic behavior for HCP detailed in Section \ref{gerarchia}, will be the key to prove our main results for the East model announced in the previous section (Theorem \ref{plateau} and \ref{asymptotics}).

Before giving a formal definition of HCP we start by saying that  of the main features of HCP is that time has a hierarchical
nature. There is an infinite sequence of \emph{epochs}  and inside
each epoch the time runs from $0$ to $\infty$. The HCP inside one
epoch is just a suitable coalescence process \emph{dependent} on the
label of the epoch. The overall evolution is obtained by suitably
linking consecutive epochs in the obvious way: the end (\ie the limit
$t\to \infty$) of one epoch coincides with the beginning of the next
one. The key link between the HCP we propose below and the East process
is provided by the very specific choice of the coalescence rates for
the $n^{th}$-epoch process. As it will be apparent below these rates
are expressed in terms of
suitable large deviation probabilities of the East process.

\subsection{Domains and classes}
In order to define our HCP we need to fix some notation and introduce
 some basic geometric concepts.
\begin{definition}
Given  a configuration $\s \in \O$  we say that the interval $[c,d]\sset \bbZ_+$,
$c<d$,  is a \emph{domain} of  $\s$ if  $\s(x)=1$ for any $x$, $c<x<d$, and $\s(c)=\s(d)=0$. If $\s\in \O_\L$ for some finite or infinite interval $\Lambda\subset\bbZ_+$ then the domains of $\s$ are defined as
the domains of the extended configuration $\tilde\s \in \O$,  given
in \eqref{estendiamo}, which are contained in $[a,b+1]$. In
particular, if $-\infty<a\leq b <\infty$,  the domains of $\s$ are
the finite intervals $[c,d]$ where $\s$ appears as $0,1,\dots, 1,0$
as well as the interval $[u, b+1]$, where $u$ is the rightmost zero
of $\s$. Given a domain $[c,d]$ its length is defined as $d-c$,
while given a zero (empty site) $x$  of $\s$ the length of the
domain having $x$ as left extreme is denoted by $d_x$.
 \end{definition}
Next we
partition $\bbN$  by the sets
$\cC_n$ defined by
\begin{equation}
\cC_0 =\{1\} \, \qquad \cC_n=[2^{n-1}+1, 2^n]  \;\text{  for }n \geq 1\,.
\end{equation}
and we set
$$\cC_{\geq n}:= \bigl(\cup_{m\geq n} \cC_m
\bigr)\,,\qquad  \cC_{>n}:=  \bigl(\cup_{m> n} \cC_m
\bigr)\,,\qquad \cC _{\leq n} := \cup_{m=0}^n \cC_n\,.
$$
\begin{definition}
Given a configuration $\s$ in $\O$ or $ \O_\L$, we say that a
domain of $\s$ is of \emph{class $n$} (respectively, at least $n$,
larger than $n$, at most $n$)  if its length belongs to $\cC_n$
(respectively, $\cC _{\geq n}$, $\cC _{>n}$, $\cC_{\leq n}$). We
also say that a zero (empty site) $x$ of $\s$
 is of class $n$ if $d_x\in \cC_n$. Similar  definitions hold  for
 $\cC _{\geq n}, \cC_{>n}, \cC_{\leq n}$.
\end{definition}
Finally,
 we point out a simple property of the sets $\cC_n$, which will be
crucial in our investigation:
\begin{equation}\label{crociata}
d, d' \in \cC_n\; \Rightarrow\; d+d'\in \cC_{>n } \,,\qquad \forall
n\geq 0\,.
\end{equation}

 In what follows we introduce (hierarchical) coalescence processes as jump stochastic dynamics on $\O$, where jumps correspond to filling an empty site.  The term coalescence is justified. Indeed, a configuration $\s\in \O$ is univocally determined by the set $\cZ(\s)$ of its zeros.
  Filling the empty site $x$ in a configuration  $\s$  corresponds to removing the point $x$ from the set $\cZ(\s)$.  Since a domain is simply the interval between consecutive zeros in $\cZ(\s)$, removing the point $x\in \cZ(\s)$ corresponds to the coalescence of the domains  on the left and on the right of $x$.

\subsection{The $n^{th}$-epoch coalescence process}\label{mammaorsa}
 We describe  here the one--epoch coalescence
process associated to the $n^{th}$-epoch (shortly $n^{th}$-CP),
which depends also on the parameters $q,\e \in (0,1)$. Fixed these
parameters, we define
 \begin{equation}
 \label{Tn}
T_0:=q^{(1-\epsilon)/2}\,,\qquad T_1 :=1/q^{3\epsilon}\,,
\qquad    T_n := (1/q) ^{(n-1)(1+3\e)}~~{\mbox{for}}~~n \geq 2\, .
  \end{equation}
Then for each $n\geq 0$  we define  the function
$\lambda_n:\bbN\to [0,\infty) $ as
\begin{equation}
\label{deflambda}
\l_n (d): =
\begin{cases}
-T_n^{-1}
\log\left( \bbP_{0\mathds{1}}^{[0,d-1]}\left( \sigma_s(0)=0 \,\forall  s\in[0,T_n]\right)\right) & \mbox{if } d \in \cC _n\,,\\
0 & \mbox{otherwise}\,.
\end{cases}
\end{equation}
where, we emphasize, $\bbP_{0\mathds{1}}^{[0,d-1]}\left(\cdot\right)$ refers to the East process in
$\L=[0,d-1]$, starting from the configuration $\s_{0\mathds{1}}$ and evolving
with parameter $q$.

Finally  we write $\O^{(\ge n)}$ for the set of configurations in  $  \O$  whose domains are all of class at least $n$. Then the $n^{th}$-CP is
 a Markov process with  paths in the Skohorod space $D([0,\infty ), \O^{(\ge n)})$ whose  infinitesimal generator $\cL_n$ acts on local functions as
\begin{equation}\label{lamamma}
 \cL _n f (\s)= \sum _{x\in \bbZ_+:\s(x)=0}  \l_n(d_x) \bigl( f( \s^x)-f(\s)\bigr)\,.
\end{equation}
Above,  $\s^x$ is the configuration obtained from $\s$ by flipping its
value in $x$, i.e. by filling the empty site $x$ (we refer to the case
$\s(x)=0$).  We will write $\bbP^{C,n}_\s$ for the law of the
$n^{th}$-CP starting from the configuration $\s$.

As observed  in \cite{FMRT},  for almost all random paths   $\{\s_s\}_{s\geq 0}$ of  the $n^{th}$-CP,
 the asymptotic configuration   $\s_\infty $ defined as
 $\s_\infty(x)=\lim _{s\uparrow \infty} \s_s(x)$ exists and it belongs
 to $\O^{(\ge n+1)}\subset \O^{(\ge n)}$.
  Hence, in what follows,  trajectories of the $n^{th}$-CP will be thought of up to time $t=\infty$ included.

\subsection{The hierarchical coalescence process (HCP)}  Fix the
parameters $q,\e\in (0,1)$.
\begin{definition}
The HCP starting from the configuration $\s\in \O $ is  the stochastic process whose evolution is described by a sequence
 of random paths
$$(\s^{(n)} _s \,:\, s\in [0,\infty])_{n \in \bbZ_+} \in D([0,+\infty], \O)^{\bbZ_+} \,,$$
such that (inductively over $n$)
$\{\s^{(n)} _s\}_{s\ge 0}$ is a random path of the $n^{th}$-CP starting from
$\s$ if $n=0$ and from $\s^{(n-1)}_\infty$  if $n\geq 1$.
\end{definition}
If the initial configuration (\ie at time $t=0$ in the first epoch)
has law $Q$ then the corresponding law and expectation
for the HCP will be denoted by $\bbP^H_Q$ and $\bbE^H_Q$ respectively.  In Section \ref{coppietta}
below we present a refined graphical construction, allowing to define on the
same probability space all the HCP as the initial configuration varies
in $\O$.
\begin{remark}
We point out that the HCP defined here  corresponds to the one in
\cite{FMRT} with the choice $\lambda^{(n)}_\ell :=0$,
$\lambda^{(n)}_r:=\lambda_{n-1}$ and   $d^{(n)}=2^{n-2}+1$ for
$n\geq 2$, $d^{(1)}=1$. Note that the index of epochs in the formulation of \cite{FMRT} runs over $\bbN$ while here it runs over $\bbZ_+$ (the HCP process $(\xi^{(n)} _s \,:\, s\in [0,\infty])_{n \in \bbN} \in D([0,+\infty], \O)^{\bbN} \,$ defined in \cite{FMRT} and with the former choices of the rates is such that, for any $n\in\bbN$ and $s\in[0,\infty]$, $\xi^{n}$ has the same law of $\sigma^{n-1}_s$).
\end{remark}
One can define the HCP in the finite volume $\L=[a,b]$ as the process whose evolution is described by a sequence of random paths
$$( \s^{(n)} _s \,:\, s\in [0,\infty])_{n \in \bbZ_+} \in D(\bar \bbR_+, \O_\L)^{\bbZ_+} $$
obtained by observing in the interval $\L$ the infinite volume
HCP  starting at the configuration  $\tilde \s$ defined in
\eqref{estendiamo}. The corresponding law with initial distribution
$Q$ will be denoted by $\bbP^{\L,H}_Q$.

\subsubsection{Graphical construction}
\label{coppietta}
As for the East process we describe a graphical construction of the
HCP in finite volume. A similar construction holds also in the
infinite volume case again using the results in \cite{Durrett,Liggett2}.

Given an interval $\L\sset \bbZ_+$  we associate to each $x\in \L$ and to each $n\in \bbZ_+$ a  Poisson process  of parameter one  and, independently, a family $\{S^{(n)}_{x,k}:k\in
\bbN\}$ of
independent random variables uniformly distributed on $[0,1]$. We assume independence as $x$
and $n$ vary. The occurrences of the Poisson process associated to the
pair $(x,n)$ will
be denoted by $(t^{(n)}_{x,k}:k\geq 0)$. The above construction defines the
probability space whose probability measure is denoted by $\bbP$. The construction of the path $(\s^n_t:t\in
[0,\infty])_{n\geq 0}$ of the HCP with initial condition $\s$  then
proceeds by induction on $n$. Set $\s^{(0)}_0=\s$. At each time
$t=t^{(0)}_{x,k}$, $k=1,2,\dots$,  if $\s^{(0)}_{t-}(x)=0$ and if $S^{(0)}
_{x,k}\leq \l_0 (d_x(\s^{(0)}_{t-}))$ then the configuration
$\s^{(0)}_t$ is obtained from $\s^{(0)}_{t-}$ by filling the site $x$.
In this  case, the occurrence $t^{(0)}_{x,k}$ is called \emph{legal
  ring}. Otherwise $\s^{(0)}_t:= \s^{(0)}_{t-}$. Clearly the limiting configuration $\s^{(0)}_\infty$ is well
defined a.s.. The path $(\s^{(1)}_t )_{t\in [0,\infty]}$ is then defined
exactly in the same way by replacing the initial configuration $\s$ with $\s^{(0)}_\infty$. The
construction is then repeated inductively.


\subsubsection{Characteristic time scales}
\label{prelimHCP}
Before moving on with the main results for the HCP we pause for a moment and establish some quantitative bounds on the characteristic
time scales of the process. Although such results are completely irrelevant for the asymptotic as $n\to \infty$ of the HCP they will play a crucial role when we will compare the HCP with the East process.
\begin{Lemma}
\label{rates bounds}
Fix $N\in \bbN$. Consider the
$n^{th}$-CP  with parameters $q$  and
$\e:=1/8N$ in the definition \eqref{Tn}. Then there exists a finite constant $c=c(N,L)$
such that, for any $n\le N$,
\begin{equation}
  \label{rbound}
  \frac{c}{t_n} \leq \min_{d\in \cC_n}\l_n(d)\leq \max_{d\in
    \cC_n}\l_n(d)\le \frac{1}{c\, t_n}\,.
\end{equation}
\end{Lemma}
\begin{proof}
The proof is based on the results of Section \ref{prelim-East} and it
can be skipped on a first reading.

Let $\tilde \tau$ be the
hitting time of the set $\{\s(0)=1\}$ for the East process with parameter
$q$ and let $\L=[0,d-1]$ with $d\in\cC_n$. Fix $n\le N$ , then
using \eqref{hit2} we get
\begin{equation}\label{boundd} \bbP_{01}^{\L}(\tilde \t \ge T_n) =
1-\bbP_{01}^{\L}(\tilde \t \leq T_n)\geq
 1-c\,T_n/t_n
\end{equation}
and the second half of \eqref{rbound} follows. In order to prove the lower bound we write
\begin{gather}
  \bbP_{0\mathds{1}}^{\L}(\sigma_s(0)=0 \,\forall  s\in[0,T_n])=\bbP_{0\mathds{1}}^{\L}(\tilde\tau \geq T_n)\nonumber\\
= \bbP_{0\mathds{1}}^{\L}(\{\tilde\tau \geq
T_n\}\cap\{\s_{T_n}=\s_{0\mathds{1}}\}) +
\bbP_{0\mathds{1}}^{\L}(\{\tilde\tau \geq T_n\}\cap\{\s_{T_n}\neq
\s_{0\mathds{1}}\})\,. \label{same2}
\end{gather}
Notice that, as $q\downarrow 0$, the  the second one is $O(q)$ thanks to Lemma
\ref{zeri}, thus the first term must be of order $O(1)$ since \eqref{boundd} guarantees that their sum is of order $O(1)$.
Moreover,
thanks to the Markov property and to (iv) of Lemma \ref{prop:hitting},
for any $t\geq t_{n+1}$ such that
$t/T_n\in \bbZ_+$,
\[ \bbP_{0\mathds{1}}^{\L}(\{\tilde\tau \geq T_n\}\cap\{\s_{T_n}=\s_{0\mathds{1}}\})^{t/T_n}\le \bbP_{0\mathds{1}}^{\L}(\tilde\tau \geq t)\le
\frac{1}{cq}\nep{-ct/t_n}\,.
 \]
Hence
\[
\bbP_{0\mathds{1}}^{\L}(\tilde\tau \geq T_n)^{t/T_n} \le \frac{1}{cq}\nep{-ct/t_n}\left(1+c'q\right)^{t/T_n}\le \frac{1}{c'' q}\nep{-c'' t/t_n}
\]
\ie
\begin{equation}
  \label{same2.1}
\bbP_{0\mathds{1}}^{\L}(\tilde\tau \geq T_n) \le \nep{-c''T_n /t_n}
\end{equation}
thus proving the first half of \eqref{rbound}.

\end{proof}

\begin{Corollary}
Fix $N\in \bbN$. Then there exists a finite
constant $c=c(N,L)$ such that the following holds. For any
$0\le n\le N$ consider the
$n^{th}$-CP in the interval $\L=[0,L-1]$ with parameters $q$  and
$\e:=1/8N$ in the definition \eqref{Tn}. Then for any $\s\in \O^{(\ge n)}_\L$

(i)
for any $x,y\in \L$ satisfying  $y-x\in
{\mathcal{C}}_n$
\begin{equation}
\label{uno} \bbP_{\sigma}^{\L,n,C}(\{x,y\}\sset \cZ(\s_t))\leq \exp(-ct/t_n)
\end{equation}

(ii)
\begin{align}
\bbP_{\sigma}^{\L,n,C} \left( | \cZ(\s)\setminus
\cZ(\sigma_{t})|\geq 1 \right) &\leq c^{-1} \,t/t_n\,, \label{i1}\\
\bbP_{\sigma}^{\L,n,C} \left( | \cZ(\s)\setminus
\cZ(\sigma_{t})|\geq 2 \right) &\leq c^{-1} \bigl(t/t_n\bigr)^2\,.
\label{i2}
\end{align}
\label{sameforHC}
\end{Corollary}

\begin{remark}
\label{remhcp} In  particular, with probability tending to one as
$q\downarrow 0$, for the $n^{th}$-CP starting from $\s\in \O^{(\ge
n)}_\L$ before time $t_n^-$ no zero has disappeared yet while after
time $t_n^+$ all the zeros of class $n$ have disappeared and
therefore the infinite time configuration has been reached (namely
for any $s\geq t_n^+$ it holds $\lim_{q\downarrow 0}
\bbP_{\sigma}^{\L,n,C}(\sigma_s=\sigma_{\infty})=1$).
\end{remark}
\begin{proof}
(i) If $\{x,y\}\sset
\cZ(\s_t)$, then the same holds at time $t=0$
and there is no extra zero between $x,y$ since otherwise we
would have a zero of class smaller than $n$ at $t=0$. Conditionally on $y\in \cZ(\s_t)$, the event
$x\in \cZ(\s_t)$ implies that the legal ring at $x$ of the
graphical construction has occurred after
time $t$. Since such a ring is an exponential
variable of parameter $\l_n (y-x)$ we conclude that
\begin{gather*}
\bbP_{\sigma}^{\L,n,C}(\sigma_t(x)=0,\sigma_t(y)=0)\leq\exp(-t\lambda_n(y-x))
\end{gather*}
and the sought bound follows from Lemma \ref{rates bounds}.

(ii) Thanks to the graphical construction, in order for $\sigma_t$ to be obtained
from $\s$ by killing at least two zeros, it is necessary that at least
two out of at most $L$ independent Poisson clocks, each one of rate smaller or
equal than $\bar\lambda_n:=\sup_{d\in{\mathcal{C}}_n}\lambda_n(d)$,
have been able to ring before time  $t$. This observation, together with Lemma \ref{rates bounds},
leads to the bound
\begin{equation}
\label{atleastone}
\bbP_{\sigma}^{\L,n,C} \left( | \cZ(\s)\setminus
\cZ(\sigma_{t})|\geq 2 \right) \leq c' \left(1-\nep{-t\bar\l_n}\right)^2 \le c''\bigl(t/t_n\bigr)^2
\end{equation}
Similarly one proves \eqref{i1}
\end{proof}

\subsection{Limiting behavior of HCP}\label{gerarchia}

In this section we recall some  asymptotic results as $n\to \infty$
obtained in \cite{FMRT} for the law of $\s^{(n)}_0$ starting from
  $Q=\text{Ren}(\nu,\mu)$.

The first result (see Theorem 2.13 in \cite{FMRT}) says that for any $n\in \bbZ_+$ and $t \in [0,\infty]$ the
 law $Q^{(n)}_t$   of  $\s^{(n)}_t$ is of the same type, i.e
$Q^{(n)}_t= \text{Ren}(\nu^{(n)}_t,\mu^{(n)}_t)$
 for suitable probability measures $\nu^{(n)}_t$ on $\bbZ_+$ and
 $\mu^{(n)}_t$ on $\bbN$.

The second result characterizes inductively
 $\nu^{(n)}:= \nu^{(n)}_{t=0}$ and $\mu^{(n)}:= \mu^{(n)}_{t=0}$
 (notice that $\nu^{(0)}=\nu$, $\mu^{(0)}=\mu$). These laws are the
 law of the first zero and the law of the domain length at the
 beginning of the $n^{th}$-epoch respectively.

Let,
 for $s \geq 0 $ and $n\in \bbZ_+$,
 \begin{equation}
   \label{eq:2}
   G^{(n)}(s) = \sum _{x \in \bbN } e^{-sx}\mu^{(n)} (x)\,,\quad
 H^{(n)}(s) = \sum _{x \in \cC_n}  e^{-sx}\mu^{(n)} (x)\,,\quad
 L^{(n)}(s) = \sum _{x \in \bbZ_+ } e^{-sx}\nu^{(n)} (x)\,.
 \end{equation}
Then  the Laplace transforms $G^{(n)}, H^{(n)}$ and $L^{(n)}$ satisfy
\begin{align}
&  1- G^{(n+1)}(s)=(1- G^{(n)}(s) ) e^{H^{(n)} (s)} \,, \label{ripetitivo1}\\
&  L^{(n+1)}(s)=
L^{(n)}(s)\exp\left(H^{(n)}(s)-H^{(n)}(0)\right)\,.
\label{ripetitivo2}
\end{align}
Finally the main result of \cite{FMRT} can be formulated as follows.
Define the rescaled variables $\tilde{X}^{(0)}:=X^{(0)}$, $\tilde Y^{(0)}:=Y^{(0)}$ and
 $$\tilde{X}^{(n)}:= X^{(n)} / (2^{n-1}+1)\,,  \qquad \tilde{Y}^{(n)}:= Y^{(n)} / (2^{n-1}+1),\quad n\ge 1 \,.
$$
where $X^{(n)}, \ Y^{(n)}$ have law $\mu^{(n)}$ and $\nu^{(n)}$
respectively.

\begin{Theorem}[\cite{FMRT}]
\label{anguriona}
Let $Q=\rm{Ren}(\nu,\mu)$ and assume that the limit
\begin{equation}\label{lim_arrosto}
c_0:=\lim _{s\downarrow 0}  \frac{-s\,{G^{(0)}}'(s)}{1-G^{(0)}(s)}\end{equation}
 exists (and then necessarily $c_0\in [0,1]$).
Assumption \eqref{lim_arrosto}  holds if: a) $\mu
  $ has finite mean and then $c_0=1$ or
b) for some $\a\in (0,1)$ $\mu$ belongs to the domain of attraction of an
$\a$--stable law or, more  generally,  $\mu\bigl((x,\infty)\bigr)=x^{-\a} L(x)$
where $L(x)$ is a slowly varying function at $+\infty$, $\a\in [0,1]$, and in
this case $c_0=\alpha$.

Then:
\begin{enumerate}[(i)]
\item The rescaled random variable   $\tilde{X}^{(n)}$  weakly
 converges  to the random variable $\tilde{X}^{(\infty)}_{c_0}$ (see
the discussion right after remark \ref{anyk})
whose Laplace transform is given  by
\begin{equation}\label{macedonia3}
\bbE\bigl(e^{-s \tilde{X}^{(\infty)}_{c_0}}\bigr)=1- \exp\Big \{ - c_0 \int _1^\infty
\frac{e^{-sx}}{x} dx \Big\}= 1-\exp\Big \{ - c_0 \, \text{Ei}(s)
\Big\},\;\; s\geq 0 \,.
\end{equation}
\item The rescaled random variable
$\tilde{Y}^{(n)}$ weakly converges to the random
 variable $ \tilde{Y}^{(\infty)}_{c_0}$,
whose Laplace transform is given  by
\begin{equation}\label{thenero}
 \bbE\bigl( e^{-s\tilde{Y}^{(\infty)}_{c_0}} \bigr)=\exp \left\{-c_0 \int_0^1
\frac{1-e^{-sy} }{y} dy \right\}\,,\; \; s\geq 0\,.
\end{equation}

\item If $Y^{(n)}$ denotes the leftmost point in $\s^{(n)}_0\,$,
 then
\begin{equation}
\mathbb{P}^H_{Q} \bigl(  Y^{(n)}= Y^{(0)}\bigr)=
 1/ (2^{n-1}+1)^{c_0(1+o(1))}\,, \end{equation}
where $o(1)$ denotes an error going to zero as $n\rightarrow \infty$.
\item  If  $Q=\text{Ren}(\mu\tc 0)$, then
\begin{equation}
\nu^{(n)}(0)=\mathbb{P}_{Q}^H \Big( \sigma_0^{(n)}(0) =0\Big)
= 1/(2^{n-1}+1)^{c_0(1+o(1))} \,.
\end{equation}
\item If furthermore $\mu$ has finite $k$-th moment then for any function $f:[0,\infty)\to\bbR$ such that $|f(x)|\leq C+Cx^k$ for some constant $C$, it holds
$$\lim_{n\to\infty}\bbE[f(\tilde X^{(n)})]=\bbE[f(\tilde X^{(\infty)}_1)]$$
\end{enumerate}
\end{Theorem}

\subsection{East process and HCP: approximation results as $q\downarrow 0$}
\label{mainresults}
The main result of this section, Theorem \ref{finitevolume} below, states that, as $q\downarrow 0$, the behavior of the
 East process  on a finite volume $\L$ and up to time $T_N= (1/q) ^{(N-1)(1+3\e)}$ (recall \eqref{Tn}), is well approximated by
the HCP having the same initial distribution.
Our second result (Theorem \ref{bacio}) states that the same occurs
for the position of  the first $k$ zeros when working in $\bbZ_+$.

Recall Definition \ref{stalling-active} of the active and stalling
periods and define, for any $t>0$, $n(t)$
and $\tau(t)$  by
\begin{equation}\label{scaletta}
 t \in [t_{n(t)}^-, t_{n(t)+1}^-)\,, \qquad  \t(t):= t-t_{n(t)}^-\,.
 \end{equation}
That allows us  to define a canonical map
\begin{equation} \label{aiutami}
\phi :D([0,\infty), \O) ^{\bbZ_+} \mapsto D([0,\infty), \O)\,,
\end{equation}
by
$$
\phi\Bigl(\bigl(\{\s^{(n)}_s\}_{s\ge 0}\bigr)_{n\in
\bbZ_+}\Bigr)_t:= \s ^{n(t)}_{\t(t)}.
$$
In the sequel and for notation convenience we will write $\s^H_t$ for the more cumbersome $\phi\Bigl(\bigl(\{\s^{(n)}_s\}_{s\ge 0}\bigr)_{n\in \bbZ_+}\Bigr)_t$ and, if confusion does not arise, we will denote by $x_k^H(t)$ the $k$-th zero of $\s^H_t$ .
In order to have compact formulas we introduce the following
convention. For $i=1,2$ let $F_i$ be a random variable with values in
some set $E$ (the same for $i=1,2$) on some
probability space $(\Theta_i,\cF_i, \bbP_i)$. Then we  define
$$ d_{TV}( \{F_1,\bbP_1\};\{F_2, \bbP_2\}):= d_{TV} (\mathfrak{p}_1, \mathfrak{p}_2 )\,,
$$
where  $\mathfrak{p}_i$ denotes the law of $F_i$ and $d_{TV}(\cdot,
\cdot)$ denotes the total variation distance.

\begin{Theorem}
\label{finitevolume}
For any $N\in \bbN$ let $\epsilon_N:=1/8N$ and choose the parameter
$\epsilon$ appearing in Definition \ref{stalling-active} and
in \eqref{Tn} equal to $\epsilon_N$. Then, for any finite interval
$\L$ and any probability measure $Q$ on $\O_\L$,
\begin{equation}
 \lim _{q\downarrow 0}\, \sup _{t \in [0, t_N^+]} \, d_{TV}\bigl(
 \{\s_t ,\,  \bbP^{\L }_Q\}  ;\, \{\s^H_t ,\,
  \bbP^{\L,H}_Q\} \bigr)=0 \,.\\
\end{equation}
\end{Theorem}
The next theorem gives the approximation result for the laws of the first zero $x_0(t)$ and of the domain interval $x_1(t)-x_0(t)$ of the East process in $\bbZ_+$ up to times $T_N$.
\begin{Theorem}\label{bacio}
Let   $Q=\text{Ren}(\nu,\mu)$ with $\mu$ such that, for any $n\ge 1$, $\mu\left([n,\infty)\right)>0$.
In the same assumption of Theorem \ref{finitevolume} and for any $k\ge 0$
\begin{equation}
\lim _{q\downarrow 0}\, \sup _{t \in [0, t_N^+]} \,
 d_{TV}\bigl(\, \{(x_{0} (t), \dots, x_k(t) ) ,\, \bbP _Q\}\,;\, \{(x^H_{0} (t), \dots, x^H_k(t) ) ,\, \bbP^H _Q\}\bigr)= 0 \,.
\label{bacio2}
\end{equation}
Assume that the $(m+\d)^{th}$-moment of $\mu$ and $\nu$ is finite for some $\d>0$. Then
\begin{equation}
\label{momentaccio1}
\lim _{q \downarrow 0} \sup_{t\in [0, t_N^+] } \bigl| \bbE_Q \bigl(
[x_{k+1} (t)-x_{k} (t)]^m \bigr)-\bbE_Q^\text{H}  \bigl( [x^H_{k+1}
(t)-x^H_{k} (t)]^m \bigr)\bigr|=0 \,.
\end{equation}
and
\begin{equation}\label{momentaccio2}
\lim _{q \downarrow 0} \sup_{t\in [0, t_N^+] } \bigl| \bbE_Q \bigl(
[x_k  (t)]^m \bigr)-\bbE_Q^\text{H}  \bigl( [x^H_k
(t)]^m \bigr)\bigr|=0 \,.
\end{equation}
\end{Theorem}
\begin{remark}
The above theorem is really a corollary of Theorem
\ref{finitevolume} once we prove a strong finite volume
approximation result both for the East process and for the HCP on
$\bbZ_+$. More precisely we will show that, up to time $t_N^+$, the
law (and the moments) of the first $k$ zeros, in the limit
$q\downarrow 0$, is very well approximated by the corresponding law
(and moments) of the finite volume processes provided that the
chosen volume is large enough (see Propositions \ref{finvol1},
\ref{finvol2} in Section \ref{sec:finvol}). The assumption
$\mu\left([k,\infty)\right)>0$ for any integer $k$ is there exactly
in order to greatly simplify the proof of such an approximation.
Without it the result still holds but its proof requires more
lengthy arguments which will appear elsewhere \cite{FMRT-Cergy}.
\end{remark}

\section{Preliminary results for the low temperature East process}
\label{prelim-East}
In this section we establish some results for the low temperature East process which will be crucial to prove the approximation with HCP in the following section.
Unless otherwise specified, we
 set $\L:=[0,L-1]$ with $L\geq 1\,$ which is fixed once and for all
 and does not
 change as $q\downarrow 0$.
Let us  begin by reviewing some  known  properties (Lemma \ref{chung} and Remark \ref{remarkzero})  which will have a fundamental role in what follows.

Recall that a site $x_j\in \cZ(\s)$ is said to be of class $n$ if
$x_{j+1}-x_j\in \mathcal{C}_n=[2^{n-1}+1,2^n]$. The next combinatorial
lemma (see
 also \cite{SE1,SE2}) says that the minimal
number of extra zeros that we have to create in order to kill a zero
of class $n$ is $n$.

\begin{Lemma}\cite{CDG}
\label{chung} Consider the  East process  on
$\L:=[0,L-1]$ starting from the completely filled  configuration
$\s_{\mathds{1}}$. For $n\geq 1$ let $V(n)$ be the set of configurations
that the process  can reach under the
condition that, at any given time, no more than $n$ zeros are
present. Define
 $$ \ell(n):=\sup_{\s\in V(n)}(L-x_0)$$
where $x_0=x_0(\s)$ is the smallest element of $\cZ(\s)$.
Then $\ell(n)=2^{n}-1$ for all $L\ge 2^n-1$.
\end{Lemma}

\begin{remark}
\label{remarkzero} If instead of considering a frozen zero boundary
condition at $L$ we consider a
deterministic time-dependent boundary condition $\s_t(L)\in \{0,1\}$, the above
lemma implies that $\ell(n) \le 2^{n}-1$ for $L\ge 2^n+1$. A key consequence is the following. Let $x_j\in
\cZ(\s)$ be of class $n\geq 1$. By definition $d_j=x_{j+1}-x_j$
belongs to the interval $[2^{n-1}+1, 2^n]$. Define $T_j$ to be the
first time the site $x_j$ is filled, then at $T_j$ there must be a zero at $x_j+1$. Thus
there must exists an intermediate time $t\in [0,T_j]$ such that, at
time $t$ there are at least $n$ zeros in the interval
$[x_j+1,x_{j+1}-1]$. Not surprisingly that will force the
characteristic time scale of $T_j$ to be of the order of $1/q^n$.
\end{remark}

\subsection{Energy barriers and characteristic time scales}\label{mammalontra}
We start by establishing two results which say that, in the limit
$q\downarrow  0$ and at any given time, the probability of observing
$k$ zeros which were not present at $t=0$ is $O(q^k)$ (see Lemma
\ref{zeri}). Therefore the probability of the event
$\{\sigma_t(x)=0\}$ coincides, for $q\approx 0$, with the probability of the event that a zero
has persisted at $x$ for the whole interval $[0,t]$ (Lemma
\ref{survivallemma}).

Then we analyze the East process starting from the  special configuration
$\s_{0\mathds{1}}$ having a single zero located at the
origin. We study two important stopping times. The first, $\tilde
\tau$, is the first time that the
origin is filled (i.e.\ the zero at the origin is removed) while the
second, $\t_{\mathds{1}}$, is the hitting time of the completely filled
configuration $\s_{\mathds{1}}$.
We obtain upper and lower bounds on the characteristic time scales of
these random times (see Lemma
\ref{prop:hitting}), which are optimal in the limit $q\downarrow  0$
(Remark \ref{nonno}).
As a consequence we establish upper bounds  on the probability
 of observing at time $t$ a zero of class $n$ (Corollary \ref{lemmaclass})
  and of killing at least one or at least two zeros of class $n$ in a time interval $t$
 (Corollary \ref{subszeros} and \ref{twozeros}, respectively).
Finally we prove that $\tilde\tau$, after a proper rescaling,
  weakly converges in the limit $q\downarrow 0$ to an exponential variable of parameter $1$ (Lemma
   \ref{poisson} and Remark \ref{remarkgamma}).

\begin{Lemma}
\label{zeri} Fix $\sigma\in\O_\L\,$, $t\geq 0$ and $k\in\bbN$. Let
$V=[0,a]\sset \L$ and let
$\{y_1,\dots, y_k\}\sset V\setminus \cZ(\s)$. Let finally $\cF$ be the
$\s$-algebra generated by the Poisson processes and coin tosses in
$\L\setminus V$. Then
\begin{equation}
\label{eq2}
\bbP^\L_{\sigma}\bigl(\bigl\{y_1,\dots, y_k\}\sset \cZ(\s_t)\tc \cF\bigr)\leq q^k\,.
 \end{equation}
Moreover
\begin{equation}
\label{eq2.1}
\bbP^\L_{\sigma}\bigl(\exists\, s\le t:\ \bigl\{y_1,\dots, y_k\}\sset
\cZ(\s_t)\tc \cF\bigr)\leq at q^k\,.
 \end{equation}
The same results hold when $\L$ is replaced by $\bbZ_+$.
\end{Lemma}
\begin{proof}
We appeal to the graphical construction of Section \ref{graphical}.
Let $y_1, \dots, y_k$ be as in the lemma, labeled in increasing
order. Given $m \geq 0$,  we write $\cA_{m}$ for the event that the
last legal ring at $y_1$ before time $t$, which is well defined
because $y_1\in \cZ(\s_t)\setminus \cZ(\s)$, occurs at time
$t_{y_1,m}$. Recall that (i) at the time $t_{y_1,m}$ the current
configuration resets its value at $y_1$ to the value of  an
independent Bernoulli$(1-q)$ random variable $s_{y_1,m}$ and (ii)
that $\cA_{m}$ depends only on the Poisson processes associated to
sites $x\ge y_1$ and on the Bernoulli variables associated to sites
$x>y_1$. Hence we conclude that
\begin{gather}
  \bbP^\L_{\sigma}\bigl(\,  \{y_1,\dots, y_k\}\sset \cZ(\s_t)\tc \cF\bigr)\\
 = \bbP\Big(\cup_{m=1}^{\infty}\bigl(\cA_{m}\, \cap \{s_{y_1,m}=0\}\bigr)\,\cap\bigl\{\{y_2,\dots, y_k\}
 \sset \cZ(\s_t)\tc \cF\bigr\} \Big)\nonumber\\
=\sum_{m=1}^\infty \bbP(s_{y_1,m}=0)\bbP\left(
\cA_{m} \,\cap\bigl\{\{y_2,\dots, y_k\}\sset \cZ(\s_t)\tc \cF\bigr\}  \right)\nonumber\\
 \leq
q\,\bbP_\s^\L\left(\{y_2,\dots, y_k\}\sset \cZ(\s_t))\tc \cF\right )\,.\nonumber
\end{gather}
A simple iteration clings the proof of \eqref{eq2}.

Let us now
consider the set $N_t=\{t_1,t_2,\dots\}$ of all the occurrences up to
time $t$ of the Poisson processes in $V$ and let $|N_t|$ be its
cardinality. Conditioned on $N_t$ and $\cF$ the probability
of seeing zeros located at $\{y_1,y_2,\dots, y_k\}$ at a given time $s\in
N_t$ is bounded from
above by $q^k$ because of exactly the same arguments that led to \eqref{eq2}. Therefore
\[
\bbP^\L_{\sigma}\bigl(\exists\, s\le t:\ \bigl\{y_1,\dots, y_k\}\sset
\cZ(\s_s)\tc \cF\bigr)\leq q^k \bbE(|N_t|\tc \cF)= Lt q^k\,.
\]
because $|N_t|$ is independent from $\cF$. Thus \eqref{eq2.1}
follows. The last statement in the
lemma is trivial.
\end{proof}

\begin{Lemma}[Persistence of zeros]
\label{survivallemma} Fix $t>0$, $\sigma\in\Omega_\L$ and $k$ sites
$y_1,\dots y_k\in \L$.  Then for any event $\mathcal A$ on
$D([0,\infty), \O_\L)$
\begin{equation*}
\bbP^\L_{\sigma}(\mathcal A \cap \{\forall i \ y_i\in\cZ(\s_t)\} )
\leq kq+ \bbP^\L_{\sigma}(\mathcal A \cap \{\forall s\in[0,t] \
\forall i \ y_i\in\cZ(\s_s)\} )\,.
\end{equation*}
The same result holds when $\L$ is replaced by $\bbZ_+$.
\end{Lemma}
\begin{proof} We can bound
\begin{gather*}
\bbP^\L_{\sigma}(\mathcal A \cap \{\forall i \ y_i\in\cZ(\s_t) \})\\
\leq \bbP^\L_{\sigma}(\mathcal A \cap \{\forall s\in[0,t] \ \forall
i \ y_i\in\cZ(\s_s) \}) + \sum_{i=1}^k\bbP^\L_{\sigma}(\exists s<t:
y_i\in \cZ(\s_t)\setminus \cZ(\s_s))\,.
\end{gather*}
For the $i^{th}$-th term in the  above sum we define the stopping time
$\tau\geq 0$ as the first time such that $\sigma_{\tau}(y_i)=1$. Then
\begin{equation*}
  \bbP^\L_{\sigma}(\exists s<t: y_i\in \cZ(\s_t)\setminus
  \cZ(\s_s))\le q
\end{equation*}
because of the strong  Markov property and Lemma \ref{zeri}.
Finally, we observe that the proof for the infinite volume East
process is completely similar.
\end{proof}

We now move to the study of the following  hitting times:
\begin{align}
\tilde\tau &: = \inf \{ t \geq 0 : \sigma_t (0)=1\}\,,\\
\t_n & := \inf \{t \geq 0: |\cZ(\s_t)\setminus \{0\}|=n\}\,,\\
\t_{\mathds{1}}&:=\inf \{ t \geq 0 : \cZ(\s_t)=\emptyset\}\,.
\end{align}
Recall that $\s_{0\mathds{1}}\in \O_\L$ is the configuration with
only one zero at the origin and that $t_n=1/q^n $ (see \eqref{deftn}).
\begin{Lemma}\label{prop:hitting}\
\begin{itemize}
\item[(i)]
 If  $L=1\in \cC_0$, then
$\bbP^\L_{0\mathds{1} } ( \t_{\mathds{1} }>t)= e^{-t (1-q)}$ for all
$t\geq 0$ \ie the hitting time $\t_{\mathds{1} }= \tilde \t $ is an
exponential time of parameter $1-q\,$.
\item[(ii)] If $L\in \cC_{\geq n}$, $n\geq 0$, then
$\bbP^\L_{0\mathds{1}}(\tau_n\leq \tilde\tau\leq\tau_\mathds{1})=1
\,.$
\item[(iii)] If  $L\in \cC_{\geq n}$ and  $n\geq 1$ (i.e.\  $L\geq
2^{n-1}+1$), then there exists a positive constant $c= c(n,L,\bar
q)$,
such that  it holds
\begin{gather}
\label{hit2} \bbP^\L_{0\mathds{1}}(\tau_\mathds{1}<t)\leq
\bbP^\L_{0\mathds{1}}(\tilde\tau<t) \leq \bbP^\L_{0\mathds{1}} (
\tau_n < t ) \le c t/t_n \,,\qquad \forall t \geq 0\,.
\end{gather}
\item[(iv)] If  $L\in \cC_{ n}$ and  $n\geq 1$ (i.e.\  $L\in
[2^{n-1}+1, 2^{n}]$), then there exists a constant $c= c(n,L)$
such that, for any $\s\in \O_\L$,
\begin{equation}
\label{hit3} \bbP^\L_{\sigma}( \tilde \tau > t )\leq
\bbP^\L_{\sigma}( \tau_\mathds{1}
> t ) \leq \frac{1}{\pi_\L(\sigma)} \exp \left\{ -ct/t_n
\right\} \,,\qquad \forall t \geq 0\,.
\end{equation}
In particular
\begin{equation}
\label{hit4}\bbP^\L_{0\mathds{1} }( \tilde \tau> t )\leq
\bbP^\L_{0\mathds{1} }( \tau_{\mathds 1}> t ) \leq \frac{1}{c\,q} \exp \left\{
-c\,t/t_n \right\}\,, \qquad \forall t\geq 0 \, .
\end{equation}
\end{itemize}
\end{Lemma}
We postpone the proof of the above lemma to Section \ref{Lemma45}.
\begin{remark}\label{nonno}
For $L\in \cC_n$ it follows from \eqref{hit3} that $\tau_{\mathds
  1}, \tilde \tau$ with high
probability are smaller than $t_n^{1+\d}$, $\d>0$. However,
due to \eqref{hit2},  we also have
$$
\lim_{q\downarrow  0} \bbP^\L_{0\mathds{1}} \left( \tau_{\mathds{1}}
< t_n^{1-\delta} \right) = \lim_{q\downarrow 0}
 \bbP^\L_{0\mathds{1}} \left( \tilde \tau < t_n^{1-\delta} \right)
= 0\, .
 $$
\end{remark}
We state three useful consequences of
the previous results. The first one (see Corollary \ref{lemmaclass}
below) gives an upper bound on the probability of seeing a domain of
class smaller or equal than $n$ at any fixed time independent of
$n$. The second and third one (see Corollary \ref{subszeros} and
Corollary \ref{twozeros}) upper bound the probability that one zero
or at least two zeros disappear in a fixed time interval when the
initial configuration has only zeros of class at most  $n$.
\begin{Corollary}[Domain survival probability]
\label{lemmaclass}
Fix $n\geq 1$. Then there exists a positive
constant $c= c(n,L)$ such that, for any $\sigma\in\O_\L$ and
any $x,y\in \L$ with $x<y$ and $y-x\leq 2^n$,
$$\bbP^\L_{\sigma}(\{x,y\}\sset  \cZ(\s_t))\leq \frac{1}{\min(q,1-q)^{2^{n}}} \exp \left\{ -c t/t_{n}\right\}
+2q\,, \qquad \forall t\geq 0\,.$$
\end{Corollary}
\begin{proof}
By using Lemma \ref{survivallemma} we get
$$\bbP^\L_{\sigma}(\{x,y\}\sset  \cZ(\s_t))\leq 2q+\bbP^\L_\s( \{x,y\}\sset  \cZ(\s_s)~\forall s\in[0,t])\,.$$
Let $\bar\sigma:=\sigma_{[x,y-1]}$. Then Lemma \ref{oriented} and (iv)
of Lemma \ref{prop:hitting} imply that
\begin{gather*}
\bbP^\L _\sigma \left(\{x,y\}\sset  \cZ(\s_s)~\forall s\in[0,t]\right)  \leq
\bbP^{[x,y-1]}_{\bar \sigma} \left(  \tilde \tau >t \right) \\
\leq  \frac{1}{\min(q,1-q)^{y-x}} \exp \left\{ - c t/t_{n}\right\}
\,.\qedhere
\end{gather*}
\end{proof}

\begin{Corollary}[Killing at least one  zero  of class at least $n$]
 \label{subszeros}
Fix $n\ge 1$. Consider the East process on
$\L$ starting from a configuration $\s$ with a zero at $x\in \L$ of
class $n_x \ge n$.
Then there exists a positive constant $c= c(n_x)$
such that
 \begin{equation*}
\bbP^\L_{\sigma}\left(x\notin  \cZ(\sigma_{t})  \right)
\leq ct/t_n
\,, \qquad \forall t\geq 0\,.
\end{equation*}
The same result holds with $\L$ replaced by $\bbZ_+$.
\end{Corollary}
\begin{proof}
Let $[x,y]$ be the domain in $\s$ whose left boundary is $x$. For
simplicity of notation we restrict ourselves to the case $y\le L-1$. Let
$\cF$ be the $\s$-algebra generated by the Poisson processes and coin
tosses associated to $[y,L-1]$. Then, thanks to
Remark \ref{remarkzero}, in order to remove $x$ within time $t$
there must exists a time $s\le t$ such that $|\cZ(\s_s)\cap (x,y)|\ge
n$. At this point the thesis follows from Lemma \ref{zeri}. The
infinite volume case is similar.
\end{proof}

\begin{Corollary}
\label{twozeros} (Killing at least two zeros of class at least $n$)
Fix $n\geq 1$ and consider the East process on
$\L$ starting from a configuration $\s$ with two zeros $x<y$ each of class
$n$. Define $\t_n(x)$ to be the first time such that in the
interval $(x,x+d_x)$ there are $n$ zeros. Similarly for $y$. Then there exists a
constant $c=c(n)$ such that
\begin{equation}
  \label{2zeros1}
  \bbP^\L_\s(\t_n(x)\le t \ ,\ \t_n(y)\le t) \le c \left(t/t_n\right)^2 .
\end{equation}
The same for the infinite volume case $\L\to \bbZ_+$.
In particular, if the initial configuration $\s$ has only zeros of
class at least $n$ then
\begin{gather*}
\bbP^\L_{\sigma}\left( \{\cZ(\sigma_{t}) \subset \cZ (\s)\}\cap
 \{|\cZ(\s)\setminus \cZ(\sigma_{t})|\geq 2\}\right)
\leq c' \left(t/t_n\right)^2\,.
\end{gather*}
for some constant $c'=c(\L,n)$.
\end{Corollary}
\begin{proof}
Call $z$ the next zero of $\s$ immediately to the right of $x$. Of
course it is possible that $z=y$ and by assumption $z=x+d_x$ with $d_x
\ge 2^{n-1}+1$. Let $\cF$ be the $\s$-algebra of the
Poisson processes and coin tosses associated to sites in
$[z,L-1]$. Then
\[
  \bbP^\L_\s(\t_n(x)\le t \ ,\ \t_n(y)\le
  t)=\bbE\left(\id_{\{\t_n(y)\le t\}}\bbP^\L_\s(\t_n(x)\le t \tc \cF)\right) \le  c(t/t_n)^2
\]
because of \eqref{hit2}.
Similarly when $\L$ is replaced by $\bbZ_+$. The second conclusion is
now immediate once we appeal to Remark \ref{remarkzero}. Indeed, in order to remove within time $t$ two zeros
$x,y$ of
class at least $n$, their respective stopping times $\t_n(x),\t_n(y)$
must have occurred before time $t$.

\end{proof}
Before moving on we summarize the overall picture that emerges from
the above results into a single proposition.
\begin{Proposition}
\label{East big picture}
Consider the East process on $\L$ starting from a
  configuration $\s$. Fix a large integer $N$, fix $\epsilon\leq
  (8N)^{-1}$ and let $t_n^{\pm}:= (1/q)^{n(1\pm\epsilon)}$.
Notice that $t_n^-\gg t_{n-1}^+$ for $n\leq N$ and  small
  $q$. Then the following picture up to scale $N$ holds with probability tending
  to one as $q\downarrow
  0$.
  \begin{itemize}
  \item The typical time necessary to kill a zero of class $n$ is of order $t_n
^{1+o(1)}$ so that at time $t_n^-$ the
zeros of $\s$ of class at least $n$ are still present.
\item At time $t\geq t_n^+$ all zeros are of class at least $n+1$ (and thus at time $t \geq t^-_n$ of class at least $n$).
\item Split the \emph{active} period $[t_n^-,t_n^+]$ into disjoint
  sub-periods of width $T_n$ defined in \eqref{Tn}. Then in each
  sub-period at most one zero of class $n$ is killed.
 \end{itemize}
\end{Proposition}
The last property follows by a simple application of Corollary
\ref{twozeros}.

\subsubsection{Proof of Lemma  \ref{prop:hitting}}
\label{Lemma45}
We are now left with the proof of Lemma \ref{prop:hitting} and for
this purpose we need first two preliminary results. Recall that
$\pi_\L$ denotes the product Bernoulli measure on $\O_\L$ with density $1-q$.

\begin{Lemma}\label{prop:amine}
For any $A \subset \Omega_\L$, the hitting time $\tau_A = \inf  \{ t
\geq 0 : \sigma_t \in A\}$ satisfies
$$
\bbP^\L_{\pi_\L} (\tau_A >t) \leq e^{-t \gap(\cL_\L) \pi_\L(A)}\,.
$$
\end{Lemma}
\begin{proof}
It is well known (see e.g. \cite{A}) that
$\bbP^L_{\pi_\L} (\tau_A >t) \leq e^{-t \l_A}$, where
$$
\lambda_A := \inf \left\{ \cD_\L(f) : \pi_\L(f^2)=1 \mbox{ and } f \equiv
0 \mbox{ on } A \right\}\,,
$$
$\cD_\L(f)=-\pi_\L(f,\cL_{\L}f)$ being the Dirichlet form of
$f$. Since $\Var_\L(f)/\pi_\L(f^2)\ge \pi_\L(A)$ if $f\equiv 0$ on $A$, from
the definition of the spectral gap
\begin{equation}
\label{gapdef} \gap(\cL_\L):=\inf_{f:\ \pi_\L(f)=0}\;
\frac{\cD_\L(f)}{\pi_\L(f^2)}\,,
\end{equation}
it follows immediately that $\lambda_A \geq \gap (\cL_\L) \pi_\L(A)$.
\end{proof}

In order to use the former result we will need in turn a sharp lower bound on the spectral gap on finite volume.

\begin{Lemma} \label{lem:gap}
Recall that $\L=[0,L-1]$. Then, for any integers $n,L$ with $1\leq L \leq 2^n$,
$$
\gap(\cL_\L) \geq \left( \frac{q}{2} \right)^n .
$$
\end{Lemma}
\begin{proof}
Fix $n,L$ as above. We first observe that, by monotonicity of the
gap for the East model (see \cite[Lemma 2.11]{CMRT}) it is enough to
consider the case $L=2^n$. For lightness of notation we denote by
$\gamma_n$ the inverse gap on $[0,2^n-1]$ (with a frozen zero at
$L=2^n$). The result is obtained by induction, following the
bisection constrained method introduced in \cite{CMRT}.

Let $A=[0,2^{n-1}-1]$ and $B=[2^{n-1},2^n-1]$ and set
$a=2^{n-1}$. Notice that~$\L:=[0,2^n-1]=A \cup B$, $A \cap B =
\emptyset$ and that $\gap(\cL_A)^{-1}=\gap(\cL_B)^{-1}=\gamma_{n-1}$. We will
denote by $\var_\L(f)$, $\pi_\L(f)$ and $\cD_\L(f)$ the variance, mean and
Dirichlet form of $f$ on the interval $\L$. Analogous notation hold for
the same quantities restricted to the intervals $A$ and $B$.

Consider the following auxiliary "constrained
 block dynamics". The block $B$ waits a mean one exponential random time and then the current configuration inside it is refreshed with a new one sampled from $\pi_B$. The block $A$ does the same but now the configuration is refreshed only if the current configuration $\sigma$ is such that $\sigma(a)=0$. The Dirichlet form
of this new chain is
$$
\cD_{\rm block}(f)= \pi_\L \left(c_A \var_A(f) + \var_B(f) \right)
$$
with $c_A(\sigma)=1-\sigma(a)$. If $\gap_{\rm block}$ is the corresponding spectral gap then Proposition 4.4 in \cite{CMRT} gives
$$
\gap_{\rm block} \geq  1-\sqrt{1-q} \geq
\frac{q}{2}\,.
$$
Hence,
$$
\var_\L(f) \leq \frac{2}{q} \pi_\L \left(c_A \var_A(f) + \var_B(f)
\right) .
$$
Now $\var_B(f) \leq \gamma_{n-1} \cD_B(f)$ and, by construction,
(see \eqref{bersani}) $c_{x}^{B}=c_{x}^{\L}$
for any $x\in B$. Therefore
$$
\pi_\L \left(\var_B(f) \right) \leq \gamma_{n-1} \sum_{x\in B}
c_{x}^{\L} \var_x(f)
$$
which is nothing but the contribution carried by $B$ to the full
Dirichlet form $\cD_\L(f)$. As far as the block $A$ is concerned one observes that
$c_{x}^{A} \cdot c_A = c_{x}^{\L}$ for any $x \in
A$. In conclusion
$$
\var_\L(f) \leq
 \frac{2}{q} \gamma_{n-1} \cD_\L(f) \quad \text{\ie}\quad  \gamma_n \leq  \frac{2}{q}\ \gamma_{n-1}.
$$
Since
$\gamma_0=1$ the result follows immediately.
\end{proof}

\begin{remark}
The above result leads to the lower bound $gap(\cL_{\Lambda})=\left(\frac{q}{2}\right)^{\log_2(1/q)}$ when $n$ is of the order $\log 2/\log (1/q)$, namely when the length of $\Lambda$ becomes of the order of the equilibrium domains. We stress that this is not the correct scaling on this length, which is instead
$\left(q\right)^{\log_2(1/q)/2}$ as proven in Section 6 of \cite{CMRT} for the lower bound and in Appendix 5 of \cite{CMST} for the upper bound.
Indeed, the above described bisection-technique should be refined as described in \cite{CMRT} to capture the correct scaling up to this length. However, for the purpose of this paper the above easier bound is enough (and this because we do analyze here the evolution at times much smaller than the global relaxation time).
\end{remark}
\begin{proof}[Proof of Lemma \ref{prop:hitting}]\\
(i) The proof  is straightforward.\\
(ii) Trivially
$\t_\mathds{1}\leq \tilde \t$. On the other
hand, if the starting configuration is $\s_{0\mathds{1} }$, Remark \ref{remarkzero} implies at once that $\t_n \leq \tilde \t$.\\
(iii) Thanks to (ii) we have
\begin{eqnarray*}
\bbP^\L_{0\mathds{1}} ( \tau_{\mathds{1}} < t ) \leq
\bbP^\L_{0\mathds{1}}(\tilde\tau<t)\leq \bbP^\L_{0\mathds{1}}( \tau_n
< t )\,.
\end{eqnarray*}
The last probability is bounded from above by $ct/t_n$ with  $c=
\binom{L-1}{n}$ by \eqref{eq2.1} in Lemma \ref{zeri}.
\\
(iv) Trivially \eqref{hit4} follows from
 \eqref{hit3}.
 The proof of \eqref{hit3}  is based  on Lemma \ref{prop:amine} with
$A=\{\s_{\mathds{1}} \}$, and on
Lemma \ref{lem:gap}. We get
\begin{gather*}
\bbP^\L_{\sigma}( \tau_\mathds{1} > t ) \leq \frac{1}{\pi_\L(\sigma)}
\bbP^\L_{\pi_\L}( \tau_\mathds{1}> t  ) \\\leq \frac{1}{\pi_\L(\sigma)}
\exp({-t \gap(\cL_\L) \pi_\L(\mathds{1} )}) \leq \frac{1}{\pi_\L(\sigma)}
\exp(-c\, t/t_n)
\end{gather*}
with $c:=(1/2)^{2^n}/2^n$.
\end{proof}
\subsubsection{Approximate exponentiality of the hitting time $\tilde \t$}
As it is very often the case for systems showing metastable behavior
(see e.g. \cite{MOS}), a loss of memory mechanism produces activation
times that become exponential variables
after appropriate rescaling. In our case the appropriate activation time is the hitting time $\tilde
\t : = \inf \{ t \geq 0 : \sigma_t (0)=1\}$.
\begin{Lemma}
\label{poisson}
Let
$f(t):=\bbP_{0\mathds{1}}^\L\left(\tilde\tau/\gamma>t\right)\,, $ where
$\gamma=\g(q,\L)$ is such that $f(1)=e^{-1}$.
Then, for any $t,s\geq 0$,
it holds
\begin{equation}
\label{poissonr} |f(t+s)-f(s)f(t)|\le c q
\end{equation}
for some constant $c>0$ independent of $q$. In particular, for any $t>0$,
$\lim_{q\downarrow 0} f(t)= \nep{-t}$.
\end{Lemma}
\begin{proof}
That  $\lim_{q\downarrow 0} f(t)= \nep{-t}$  follows from
\eqref{poissonr} by standard arguments (see
e.g. \cite{Olivieri-Vares}).  We now
prove \eqref{poissonr}. For any $s,t\geq 0$ we have
\[
f(t+s)=\bbP^\L_{0\mathds{1}}(\tilde\tau>\g(t+s) \tc\tilde\tau>\g
s)f(s).
\]
Moreover, thanks to the Markov property,
\begin{multline}
\label{ss}
\bbP^\L_{0\mathds{1}}(\tilde\tau>\g(t+s) \tc\tilde\tau>\g s)
=f(t)\,\bbP^\L_{0\mathds{1}}(
\sigma_{s\gamma}=\sigma_{0\mathds{1}}\tc \tilde\tau>\g s)\\
+
\bbP^\L_{0\mathds{1}}(\tilde\tau>\g (t+s) \tc \tilde\tau>\g s\,;\,
\sigma_{\g s}\neq \sigma_{0\mathds{1}})\bbP^\L_{0\mathds{1}}(
\sigma_{\g s}\neq\sigma_{0\mathds{1} }|\tilde\tau>\g s)\,,
\end{multline}
and therefore
\begin{gather*}
|f(t+s)-f(t)f(s)|\le    2\bbP^\L_{0\mathds{1}}(\{\sigma_{s\gamma}\neq
\sigma_{0\mathds{1}}\}\cap \{\tilde\tau>\g s\})\le 2Lq\,.
\end{gather*}
where we used Lemma \ref{zeri} in the last inequality.
\end{proof}
\begin{remark}
\label{remarkgamma} If $\L=[0,L-1]$ with $L\in\mathcal{C}_n\,$, then by using Remark \ref{nonno} we obtain
$\gamma=t_n^{1+o(1)}$ as $q\downarrow 0$.
\end{remark}

\subsection{Finite volume approximation}
\label{sec:finvol}
In this section we prove a
finite volume approximation result of the infinite volume East process with
initial distribution $Q=\text{Ren}(\n,\mu)$, provided that
$\mu\left([n,\infty)\right)>0$ for any $n$.
\begin{Proposition}
\label{finvol1}
Let $Q=\text{Ren}(\n,\mu)$ and suppose that
$\mu\left([n,\infty)\right)>0$ for any $n$. Then, for any $\ell$ and
any $N$,
\begin{equation*}
\lim _{L\uparrow \infty}\limsup_{q\downarrow 0}\sup_{t\in [0,t^+_N]}
d(t,\ell,L)=0
\end{equation*}
where $d(t,\ell,L)$ denotes the variation distance between the laws of
the vector $\left(\s(0),\dots,\s(\ell)\right)$ at time $t$ for the East process in
$\L=[0,L-1]$ and the East process in $\bbZ_+$ with initial
distribution $Q$.
\end{Proposition}
\begin{proof}
Let $n_0\ge N+1$ be such that $\mu(2^{n_0})>0$ and let
$$
\cA_L:=\{\s\in \O:\ \exists\, x_j\in \cZ(\s)\cap [\ell+1,L-2^{n_0}]
\text{ of class
  $n_0$} \}\, .
$$
From the renewal property of $Q$ it follows that $ \lim_{L\to
\infty} Q\left(\cA_L\right)=1$. For any $\s\in \cA_L$ let $x_*\in
[\ell+1,L-2^{n_0}]$ be the smallest zero of $\s$ of class $n_0$.
Since $n_0\ge N+1$ Corollary \ref{subszeros} (see also Proposition
\ref{East big picture}) implies that $x_*\in \cZ(\s_s)\ \forall s\le
t_N$ with probability tending to one as $q\downarrow 0$ both for the
finite and infinite volume East processes starting from $\s$.
Finally, conditioned to the event that for both processes $x_*$ is
not killed up to time $t_N$, the graphical construction implies that
their evolutions in $[0,\ell]$ coincide up to time $t_N$.
\end{proof}
In order to state the next result it is convenient to define, for any
$\L=[0,L-1]$, $x_k^\L:=x_k(\s)$ if $|\cZ(\s)\cap \L|\ge
k$ and $x_k^\L=L$ otherwise.
\begin{Proposition}
\label{finvol2} Let $Q=\text{Ren}(\n,\mu)$ and suppose that
$\mu\left([n,\infty)\right)>0$ for any $n$. Then for any $N\in \bbN$
and $k\ge 0$ it holds:\\
(i)
\begin{equation}
\lim_{L\uparrow \infty}\limsup_{q\downarrow 0}\, \sup _{t \in [0, t^+_N]}
\,
 d_{TV}\bigl(\, \{(x_{0} (t), \dots, x_k(t) ) ,\, \bbP_Q\}\,;\, \{(x^\L_{0} (t), \dots, x^\L_k(t) ) ,\, \bbP^\L_Q\}\bigr)= 0 \,.
\label{finvol2.1}
\end{equation}
(ii) If the $(m+\d)^{th}$-moments of
$\nu,\mu$ are both finite for some $\d>0$ then
\begin{equation}
\label{finvol2.2} \lim_{L\uparrow \infty}\limsup_{q \downarrow 0}
\sup_{t\in [0, t^+_N] } \bigl| \bbE_Q \bigl( [x_{k+1} (t)-x_{k} (t)]^m
\bigr)-\bbE_Q^\L  \bigl( [x_{k+1}^\L (t)-x^\L_{k} (t)]^m \bigr)\bigr|=0
\end{equation}
and
\begin{equation}
\label{finvol2.3} \lim_{L\uparrow \infty}\limsup_{q \downarrow 0}
\sup_{t\in [0, t^+_N] } \bigl| \bbE_Q \bigl( [x_k(t)]^m
\bigr)-\bbE_Q^\L  \bigl( [x^\L_k (t)]^m \bigr)\bigr|=0 \,.
\end{equation}
\end{Proposition}

\begin{proof} Let us fix $N$ and $n_0\ge N+1$ such that $\mu_0:=\mu(2^{n_0})>0$.\\
(i) Fix $k$ and $L\gg k2^{n_0}$, and consider the event
$$
\cB_{k,L}:=\{\s:\ \exists \ \{x_{j_i}\}_{i=0}^k\sset \cZ(\s)\cap
[0,L-2^{n_0}] \text{ such that each $x_{j_i}$ is of class $n_0$}
\}\,.
$$
Clearly $\lim_{L\uparrow \infty}Q\left(\cB_{k,L}\right)=1$. For
$\s\in \cB_{k,L}$ let $\{x_*^{(i)}\}_{i=0}^k$ be the smallest zeros
with the properties described in $\cB_{k,L}$.  Conditionally on the
event that none of the  $\{x_*^{(i)}\}_{i=0}^k$ has been killed for
both processes within time $t^+_N$, the first $k+1$ zeros of the
East process in $\L=[0,L-1]$ starting from $\s$ necessarily coincide
with those of the infinite volume East process. Since the
conditioning event has probability tending to one as $q\downarrow 0$
thanks to Corollary \ref{subszeros} (both for the East process on
$\L$ and the infinite volume East process) the thesis follows.

(ii)-\eqref{finvol2.2} For simplicity and without loss of generality
we only discuss the case $k=0$. The argument here is conceptually similar to that employed in
the proof of (i) but more involved. The reason is that, in order to
have a good control on the probability $\bbP_Q(x_1(t)-x_0(t)\ge
\ell)$, $\ell$ very large even compared to e.g. $1/q$, a single zero
at time $t=0$ of class at least $n_0$ is no longer enough. What we
really need are enough (typically $O(\log(\ell))$ such zeros in
order to be sure that with not too small probability at least one of
them has survived up to time $t$. Thus the argument will be split
into a part in which the $Q$-large deviations of the number of such
zeros dominate and a second part in which the ``resistance'' of each
zero will play a key role.

We write
\begin{gather}
  \bigl| \bbE_Q \bigl([x_1 (t)-x_{0} (t)]^m \bigr)-\bbE_Q^\L  \bigl( [x_1^\L(t)-x^\L_{0} (t)]^m \bigr)\bigr| \nonumber \\
\le \bigl| \bbE_Q \bigl([x^\L_1 (t)-x^\L_{0} (t)]^m \bigr)-\bbE_Q^\L  \bigl( [x_1^\L
(t)-x^\L_{0} (t)]^m \bigr)\bigr| \label{A}\\
+ \bigl| \bbE_Q \bigl([x_1 (t)-x_{0} (t)]^m -[x^\L_1 (t)-x^\L_{0}
(t)]^m\bigr)|\,. \label{B}
\end{gather}
Let us deal first with the term \eqref{A}.

Let
\begin{align*}
\cD&=\{\s:\ x_1(\s)\le L/2\}\\
 \cB &=\{\s:\ \exists x<y<z \in \cZ(\s)\cap [0,L] \text{ with $x\geq x_1(\sigma)$ and
  $[x,y]$, $ [y,z]$ of class at least $n_0$}  \}.
\end{align*}
Then we can split the integration w.r.t. $Q$ over the set $\cD\cap\cB$ and the set $\cD^c\cup \{\cD\cap\cB^c\}$.

For any $\s\in \cD\cap\cB$
$$
\lim_{q\downarrow 0} \sup_{t\in [0,t_N]} \bigl| \bbE_\s
\bigl([x^\L_1 (t)-x^\L_{0} (t)]^m \bigr)-\bbE_\s^\L  \bigl( [x_1^\L
(t)-x^\L_{0} (t)]^m \bigr)\bigr|=0
$$
exactly by the same argument that was used in (i). Thus
$$
\limsup_{q\downarrow 0}\sup_{t\in [0, t_N]} \bigl| \bbE_Q
\bigl([x^\L_1 (t)-x^\L_{0} (t)]^m \bigr)-\bbE_Q^\L  \bigl( [x_1^\L
(t)-x^\L_{0} (t)]^m \bigr)\bigr| \le L^m \left(\,Q(\cD^c) +
Q(\cD\cap \cB^c)\,\right)\,.
$$
The boundedness of the $m^{th}$-moment of $\mu,\nu$ implies that
$$
\lim_{L\uparrow \infty}L^mQ(\cD^c)\le \lim_{L\to
\infty}L^m\left(\nu([L/4,\infty)+\mu([L/4,\infty)\right)=0\,.
$$
Consider now the contribution $ L^m Q(\cD\cap \cB^c)$.
Note that  if we set
$d_i(\s)=x_{i+1}(\s)-x_i(\s)$  then
$ \cD\cap \cB^c\subset \cA_1\cup \cA_2$
where
\begin{align*}
& \cA_1=  \{\s:\ d_i(\s)< 2^{n_0}, \, \forall \ 1\leq i\leq
L/2^{n_0+1}\}\,,\\
&\cA_2:=\bigcup_{x=0}^{L/2} \bigcup _{\ell = 2^{n_0} }^{L-x} \bigcup _{j= 1}^{L -\ell-x}  \{\s:\
x_1(\sigma)=x, \; d_{j}(\s)=\ell , \;
 d_i(\s)<
2^{n_0}\;\forall i \not =j \text{ with } 1\leq i \leq C(x,\ell)\}
\end{align*}
where $C(x,\ell)=2^{-n_0} [L-x- \ell] $.
It is immediate to verify that $Q(\cA_1)$ is smaller than
$(1-\mu_0)^{L/2^{n_0+1}}$. Hence \[\lim_{L\uparrow
\infty}L^mQ(\cA_1)=0.\] By a union bound we get
\begin{align*}
Q(\cA_2)
& \leq\sum_{x=0}^{L/2} Q(x_1=x)
\sum  _{\ell = 2^{n_0} }^{L-x}    (L-\ell-x)\mu(\ell) \mu([1, 2^{n_0}])^{2^{-n_0} [L-x- \ell]-1}\\
& \leq \sum_{x=0}^{L/2} Q(x_1=x)\left[
\sum  _{\ell = 2^{n_0} }^{L-x-\lambda\log L} \frac{L^{1-a}}{ \mu([1, 2^{n_0}])}
+ \sum_{\ell=L-x-\lambda\log L+1}^{L-x}(\lambda\log L)\mu(\ell)\right]\\
& \leq \frac{  L^{2-a}}{ \mu([1, 2^{n_0}])}+ \lambda^2(\log L)^2\mu[L/4,\infty]
\end{align*}
where $a:=\lambda 2^{-n_0}|\log  \mu([1, 2^{n_0}])|$  and $\lambda$
is a positive constant  chosen so that $a>m+2$. Therefore by the
boundedness of the $(m+\d)$-th moment of $\mu$ we get
\[\lim_{L\uparrow \infty}L^mQ(\cA_2)=0.\]

We now examine the term \eqref{B}.

It is immediate to verify  that
\begin{gather}
\bigl| \bbE_Q \bigl([x_1 (t)-x_{0} (t)]^m -[x^\L_1 (t)-x^\L_{0}
(t)]^m\bigr)|\le \bbE_Q\left( x_1(t)^m \id_{x_1(t)\ge L}\right)\nonumber\\
\leq L^m\bbP_Q(x_1(t)\ge L)+ c_m\sum^\infty_{j=L}j^{m-1}
\bbP_Q(x_1(t)\ge j)\label{gattini100}
  \end{gather}
for a suitable constant $c_m$ depending on $m$. \\
Consider a generic term $ \bbP_Q(x_1(t)\ge j)$, $j\ge L$. Once again
we split the $Q$-integration over $\cD_j=\{\s:\ x_1(\s)\le
j/2\}$ and $\cD_j^c$ to get
\begin{equation}\label{sincronia}
  \bbP_Q(x_1(t)\ge
j)\le Q(\cD_j^c) + \int_{\cD_j} dQ(\s)\bbP_\s(x_1(t)\ge j)\,.
\end{equation}
The contribution to \eqref{gattini100} of the first term in the
r.h.s. of \eqref{sincronia} is fine, \ie
\[\lim_{L\uparrow \infty}\Big\{  L^m \, Q(\cD_L^c)+ c_m \sum_{j= L}^\infty
j^{m-1}Q(\cD_j^c)\Big\}=0\]
again because the
$m^{th}$-moment of $\mu,\nu$ is finite. We further split the last term
in \eqref{sincronia} as
\begin{gather*}
\int_{\cD_j} dQ(\s)\bbP_\s(x_1(t)\ge j)= \int_{\cD_j\cap \cB_j}
dQ(\s)\bbP_\s(x_1(t)\ge j)+\int_{\cD_j\cap\cB_j^c}
dQ(\s)\bbP_\s(x_1(t)\ge j)\\
\le \int_{\cD_j\cap \cB_j}
dQ(\s)\bbP_\s(x_1(t)\ge j)+Q(\cD_j\cap \cB_j^c)
\end{gather*}
where
\[\cB_j=\{\s:\ \exists \{x_{j_i}\}_{i=0}^\ell\in \cZ(\s)\cap [x_1(\s),j]\
\text{and each $x_{j_i}$ is of class at least $n_0$}\},\]
with $\ell= \inte{\l \,m \log(j)/\mu_0}$,
$\l$ being a numerical constant to be chosen
later on.

Let $\cN(\s)=|\cZ(\sigma)\cap[x_1(\s),j]|$ and let $\cN_0(\s)= |\{1 \leq i\le 2\ell/\mu_0 + 1:\ x_i(\s) \text{ is
  of class at least $n_0$}\}|$.
  Then
\begin{gather}
  Q(\cD_j\cap \cB_j^c)\le Q(\{\cN\le 2\ell/\mu_0\} \cap \cD_j ) +
  Q(\cN_0\le \ell)
\label{C}
\end{gather}
where we used $\{\cN> 2\ell/\mu_0\}\cap \cB_j^c \sset \{\cN_0\le \ell\}$.

 In turn, by standard binomial large
deviations,
\begin{gather*}
Q(\cN_0\le \ell)\le  \nep{-c\ell \mu_0}\le j^{-\l cm}
\end{gather*}
for some numerical constant $c$. \\
As far as the term $Q(\cN\le
2\ell/\mu_0 \,;\, \cD_j )$ is concerned we have
\begin{gather*}
Q(\{\cN\le 2\ell/\mu_0\}\cap \cD_j )\le Q(\sum_{i=1}^{\inte{2\ell/\mu_0}}d_i(\s) \ge j/2)
\le \frac{2\ell}{\mu_0}  \mu\left(\,[j/(4\ell/\mu_0)\right)\,.
\end{gather*}
In conclusion, if $\l$ is taken large enough and by using the
assumption on the finiteness of the $(m+\d)^{th}$-moment of $\mu$,
we conclude that
\[\lim_{L\uparrow \infty}\left(L^m Q(\cD_L\cap \cB_L^c) + c_m\sum_{j\ge L} j^{m-1}Q(\cD_j\cap \cB_j^c)\right)=0\,.\]
We are left with the analysis of
\[\int_{\cD_j\cap
  \cB_j}dQ(\s)\bbP_\s(x_1(t)\ge j)\leq \sup_{\s\in \cD_j\cap\cB_j}\bbP_\s(x_1(t)\ge j)\,.\]
Given $\s\in \cD_j\cap\cB_j$ let $\left\{[a_i,b_i],\ b_i\le
  a_{i+1}\right\}_{i=1}^\ell$ be the first $\ell$ domains of class
$n_0$ in $\s$, contained in $[x_1(\s),\,j]$,
whose existence is guaranteed by $\s$ being in $\cB_j$. Then
\[\{x_1(t)\ge j\}\sset \cup_{k=0}^\ell \cap_{i=1\atop i\neq k}^\ell \{\s_t(x)=1\ \forall x\in
[a_i,b_i)\}:= \cup_{k=0}^\ell \cap_{i=1\atop i\neq k}^\ell A_i\,.
\]
\ie
\[
\bbP_\s(x_1(t)\ge j)\le \bbP\bigl(  \cap_{i=1\atop i\neq k}^\ell A_i\bigr)
\]
We claim that
\begin{equation}
  \label{D}
\bbP_\s\bigl(  \cap_{i=1\atop i\neq k}^\ell A_i\bigr)\le
\begin{cases}
\b^{\ell-1}
&\text{if $k\ge 1$}\\
  \b^{\ell} &
  \text{if $k=0$}
\end{cases}\,\;.
\end{equation}
with $\b(q)=ct_N/t_{n_0}$. Notice that $\lim_{q\downarrow 0}\b(q)=0$ since
$n_0> N$. Moreover  the r.h.s. of \eqref{D} is smaller than an arbitrarily
large inverse power of $j$ for $q$ small enough because
$\ell=O(\log(j))$. Hence, assuming \eqref{D},  the proof of (ii) is
finished since
\[\lim_{q\downarrow 0}\sup_{t\in [0,t_N^+]}\left[L^m \sup_{\s\in \cD_L\cap\cB_L}\bbP_\s(x_1(t)\ge
L) +c_m\sum_{j=L}^\infty j^{m-1}\sup_{\s\in
\cD_j\cap\cB_j}\bbP_\s(x_1(t)\ge j)\right] = 0\,.
\]
For simplicity we prove the claim \ref{D} only for $k=0$ but the general case
is the same. We observe that the event
$\cap_{i=2}^\ell A_i$ is measurable w.r.t. to the $\s$-algebra $\cF$
generated by the
Poisson processes and coin tosses at the sites in $[b_1,\infty)$. Therefore
\[\bbP_\s(\cap_{i=1}^\ell A_i)= \bbE_\s\left( \prod_{i=2}^\ell
\id_{A_i}\ \bbP_\s(A_1\tc \cF)\right)\,. \] Thus, by iteration, it is enough to prove
that
\[\bbP_\s(A_1\tc \cF)\le \b\]
with $\b$ as above. That follows, as usual, from Corollary  \ref{subszeros} and Remark
\ref{remarkzero}. \\
(ii)-\eqref{finvol2.3} The proof is just a trivial adaptation of the proof of \eqref{finvol2.2}.
\end{proof}

\begin{remark}
\label{finvolHCP} The same results of Propositions \ref{finvol1} and
\ref{finvol2} hold for the hierarchical coalescence process HCP and
the proof is practically the same with one big simplification. As
soon as a zero of class bigger than $N$ occurs in the initial
configuration, then, up to time $t_N$, the zero cannot be erased.
Thus, in this case, the uniform bound on, say, $\bbE_Q(x_1(t)^m)$
are all obtained through a control of the corresponding
$m^{th}$-moment of $Q$ without any need of ``dynamical'' estimates.
\end{remark}

\section{East process and HCP: proof of Theorems \ref{finitevolume} and \ref{bacio}}
\label{sectionbacio}
As already mentioned at the end of Section \ref{mainresults}, the
proof of Theorem \ref{bacio} follows at once from Theorem
\ref{finitevolume} together with Propositions \ref{finvol1} and
\ref{finvol2} and their analog for the HCP process (see Remark
\ref{finvolHCP}). Thus the key point here is to prove Theorem
\ref{finitevolume}.

Without loss of generality we can assume that (the label of) the
largest epoch $N$ which we will observe is larger than one and that
$\epsilon=1/8N$. We will show below that the proof of the theorem
can be reduced to the proof of the following claim.
\begin{claim}
\label{mainclaim}
Let $\s\in \O_\L$ be such that any zero in $\cZ(\s)$ is of class at
least $n$. Then
\[\lim_{q\downarrow 0}\sup_{t\in [0, t_n^+-t_n^-]}d_{TV}\left(\{\s_t,\bbP^\L_\s\}\,;\,\{\s_t,\bbP^{\L,n,C}_\s\}\right)=0\]
where we recall $\bbP^{\L,n,C}_\s$ denotes the law of the
$n^{th}$-coalescence process on $\L$ starting from $\s$ and defined
in Section \ref{mammaorsa} with the choice $\epsilon=1/8N$.
\end{claim}
Let us explain how to derive Theorem \ref{finitevolume}  assuming
the claim.  Fixed $t\le t_N^+$, there are two cases to be examined:
\begin{enumerate}[(a)]
\item $t$ belongs to an \emph{active period} \ie $t\in [t_{n}^-,t_n^+]$;
\item $t$ belongs to a \emph{stalling period} \ie
$t\in [t_{n-1}^+,t_n^-]$.
\end{enumerate}
We first observe that during a stalling period nothing happens with
probability tending to one as $q\downarrow 0$. More precisely, for any
$\s\in \O_\L$, by using Proposition \ref{East big picture} we get
\begin{equation}
  \label{noaction}
  \lim_{q\downarrow 0}\sup_{n\le N}\sup_{t\in [t_{n-1}^+,t_n^-]}d_{TV}\left(\{\s_t,\bbP^\L_\s\}\,;\,\{\s_{t_{n-1}^+},\bbP^{\L}_\s\}\right)=0
\end{equation}
and similarly for the HCP process by using
 Corollary \ref{sameforHC}.

Thus, by a simple triangular inequality for the variation distance,
it is enough to consider only case (a). For this purpose we first
observe that, thanks to the Markov property, to the fact that with
probability tending to one as $q\downarrow 0$ all the zeros of
$\s_{t_n^-}$ are of class at least $n$ (Corollary \ref{sameforHC}
and Proposition \ref{East big picture}) and to Claim
\ref{mainclaim},
\[\lim_{q\downarrow 0}\sup_{s\in [0, t_n^+ -
  t_n^-]}d_{TV}\left(\{\s_{t^-_n+s},\bbP^\L_\s \}\, ; \, \{\s_{t^-_n+s},\bbP^{\L,H}_\s \}\right)=0\]
if
\[\lim_{q\downarrow 0}d_{TV}\left(\{\s_{t^-_n},\bbP^\L_\s \}\, ; \, \{\s_{t^-_n},\bbP^{\L,H}_\s \}\right)=0\,.\]
In turn, thanks to \eqref{noaction},  the above holds if
\begin{equation}\label{ricorico}\lim_{q\downarrow 0}d_{TV}\left(\{\s_{t^+_{n-1}},\bbP^\L_\s \}\, ; \, \{\s_{t^+_{n-1}},\bbP^{\L,H}_\s \}\right)=0\,.\end{equation}
If we recursively iterate the above argument (note that  when $n-1=0$ \eqref{ricorico} holds from Claim \ref{mainclaim}  since $t_0^-=0$) we get the sought conclusion. Thus the proofs of Theorem \ref{finitevolume} and \ref{bacio} are completed once we prove Claim \ref{mainclaim}.

\subsection{Proof of Claim \ref{mainclaim}}
Recall that (see \eqref{Tn})
\[T_0=q^{(1-\epsilon)/2},\quad T_1=1/q^{3\epsilon}, \quad T_n=(1/q)
^{(n-1)(1+3\e)} \text{ for $n \geq 2$}.\]
Fix $n\le N$ and divide the time interval $[0,t_n^+-t_n^-]$ into $M_n=(t_n^+-t_n^-)/T_n$ {\sl active sub--periods}
 $[t^{(\ell)}, t^{(\ell+1)})$ ($[t^{(\ell)}, t^{(\ell+1)}]$ if $\ell = M_n-1$) where $t^{(\ell)}:=\ell\, T_n$. Here we are neglecting
the integer part for lightness of notation. Thus
$M_0=q^{-(1+\epsilon)/2}$,
$M_1=\frac{1-q^{2\epsilon}}{q^{1-2\epsilon}}$ and
$M_n=\frac{1-q^{2n\epsilon}}{q^{1+3\epsilon-2n\epsilon}}$ if $n\geq
2$.

\begin{definition}[$t$-trajectories and good $t$-trajectories]
\label{goodtr}\\
Fix $t\in  [0, t_n^+-t_n^-]$ and $\s\in \O_\L$ such that all its zeros are of class at least $n$.
 Let $\cT := \{t^{(\ell)}:\ t^{(\ell)}\leq t\,,\; 0\leq \ell \leq M_n\}\cup\{t\}$.
The $t$--trajectory $\vec \s$ of a path $\{\s_s\}_{s\geq 0}\in D(
[0,\infty), \O_\L)$, such that $\s_0=\s$, is  obtained restricting $\s_s$
to $s\in \cT$. We will often write $\vec \s_\ell$ for
$\s_{t^{(\ell)}}$. A $t$--trajectory $\vec \s$ is called \emph{good}
if given two arbitrary consecutive times $s<s'$ in $\cT$ then either $\s_{s'}=\vec \s_s$ or
$\s_{s'}$ is obtained from $\vec\s_{s}$ by removing a
single zero of class $n$. The set of all good $t$--trajectories is
denoted by $\cG_t(\s)$.
\end{definition}
It follows from Corollary \ref{sameforHC} and Proposition \ref{East big picture}
that the set of good $t$-trajectories has probability tending to
one as $q\downarrow 0$ both for the East process and for the $n^{th}$-CP.
The key to prove the claim will be the following result.
\begin{Proposition} For any $\s\in \O_\L$ such that all its zeros are of class at least $n$
\label{good}
\[
\lim_{q\downarrow 0} \sup_{t\le t_n^+-t_n^-} \sum_{\vec \s \in \cG_t(\s)}|\bbP^\L_{\sigma}(\vec \sigma)-\bbP^{\L,n,C}_{\sigma}(\vec \sigma)|=0.
\]
\end{Proposition}
Assuming the proposition we conclude the proof of Claim
\ref{mainclaim} as follows. Let $\cE\sset \O_\L$ and write
\begin{gather*}
\bbP_\s^\L(\s_t\in \cE)=  \sum_{\stackrel{ \vec\s \in
\mathcal{G}_t(\s) }{\s_{t}\in \mathcal{E}} }
\bbP^\L_{\sigma}(\vec \s) + \sum_{\stackrel{ \vec\s \in
\mathcal{G}^c_t(\s) }{\s_{t}\in \mathcal{E}} }
\bbP^\L_{\sigma}(\vec \s)
\end{gather*}
and similarly for the $n^{th}$-CP. Thus
\begin{gather}
|\bbP_\s^\L(\s_t\in \cE)-\bbP_\s^{\L,n,C}(\s_t\in \cE)|\nonumber\\
\le \sum_{\vec\s \in \mathcal{G}_t(\s)}|\bbP^\L_{\sigma}(\vec
\s)-\bbP^{\L,n,C}_{\sigma}(\vec \s)| +\bbP^\L_\s(\cG_t(\s)^c)
+\bbP^{\L,n,C}_\s(\cG_t(\s)^c)\,. \label{corocoro}
\end{gather}
As observed before Proposition \ref{good}, the last two terms in the
r.h.s. tend to zero as $q\downarrow 0$ in a strong sense, namely
\[
\lim_{q\downarrow 0}\sup_{t\le t_n^+-t_n^-}
\Bigl[\bbP^\L_\s(\cG_t(\s)^c)
+\bbP^{\L,n,C}_\s(\cG_t(\s)^c)\Bigr]=0\,.
\]
The first term in the r.h.s of \eqref{corocoro} tends to zero because
of Proposition  \ref{good}. Claim \ref{mainclaim} is proved.
\\

\begin{proof}[Proof of Proposition \ref{good}] For simplicity we
restrict to times $t$ of the form $t= t^{(\ell)}$ for some $\ell
\leq M$. The general case can be treated similarly. Moreover for
lightness of notation we will drop the superscript $\L$ and $n$ from
our notation $\bbP_\s^{\L},\ \bbP_\s^{\L,n,C}$. Recall that
$\O_\L^{(\ge n)}$ denotes the set of configurations $\s$ such that
each $x\in \cZ(\s)$ is of class at least $n$ and define
\begin{align}
\d&= \sup_{\s\in \O_\L^{(\ge n)}}\max\left( \big|\, \frac{\bbP_\s(\s_{T_n}=\s)}{\bbP_\s^{C}(\s_{T_n}=\s)}-1\,\big|,
\big|\, \frac{\bbP^C_\s(\s_{T_n}=\s)}{\bbP_\s(\s_{T_n}=\s)}-1\,\big|\right) \label{delta}
\\
\g&=  \sup_{\s\in \O_\L^{(\ge n)}}\sup_{\stackrel{x\in \cZ(\s)}{x
    \text{ of class } n}}\max\left(
\big|\,\frac{\bbP_\s(\cZ(\s_{T_n})=\cZ(\s)\setminus
\{x\})}{\bbP_\s^{C}(\cZ(\s_{T_n})=\cZ(\s)\setminus \{x\})}-1\,\big|,
\big|\,\frac{\bbP^C_\s(\cZ(\s_{T_n})=\cZ(\s)\setminus
\{x\})}{\bbP_\s(\cZ(\s_{T_n})=\cZ(\s)\setminus \{x\})}-1\,\big|
\right)\,. \label{gamma}
\end{align}
Then, by the Markov property, given $\vec \sigma\in \cG_t(\s)$ it
holds
\begin{gather*}
 \Big| \frac{\bbP_{\sigma}(\vec \sigma)}{\bbP^{C}_{\sigma}(\vec \sigma)}-1\Big|=\Big|
\prod_{\ell=0}^{M_n-1} \frac{\bbP_{\vec\s_\ell}(\s_{T_n}=\vec\s_{\ell+1})}{\bbP^C_{\vec\s_\ell}(\s_{T_n}=\vec\s_{\ell+1})}-1\Big|\\
\le \left(1+\d\right)^{M_n}(1+\g)^c -1
\end{gather*}
for some constant $c$ depending on $(\L,n)$, because the number of
transitions $\vec\s_\ell \to \vec\s_{\ell+1}$ in which a zero is
removed is uniformly bounded (e.g. by the cardinality of $\L$). Above we have used \eqref{delta} and
\eqref{gamma} because $\vec\s_{\ell}\in \O^{(\ge n)}$ for all $\ell$
since $\s\in \O^{(\ge n)}$ and $\vec \sigma\in \cG_t(\s)$. Hence it is sufficient to show that
$\lim_{q\downarrow 0}M_n\delta =0 $ and $\lim_{q\downarrow
0}\gamma=0 $.

\subsubsection{Bounding $\d,\g$}
It follows from Corollary \ref{sameforHC} and Proposition \ref{East
  big picture} that
$\bbP_\s^C(\s_{T_n}=\s)\ge 1/2$ for $q$ small enough uniformly in $\s\in
\O_\L^{(\ge n)}$ and similarly for $\bbP_\s(\s_{T_n}=\s)$. Thus
\begin{gather}
  \d \le 2 \sup_{\s\in \O_\L^{(\ge n)}} \Big|\bbP_\s(\s_{T_n}=\s)-\bbP^C_{\sigma}(\s_{T_n}=\s)\Big|\nonumber\\=
  2 \sup_{\s\in \O_\L^{(\ge n)}}\Big|\bbP_\s(\s_{T_n}\neq \s)-\bbP^C_{\sigma}(\s_{T_n}\neq \s)\Big|\nonumber \\
\leq 2 \sup_{\s\in \O_\L^{(\ge n)}}\left[ \bbP_\s\bigl(\cZ(\s_{T_n})\nsubseteq \cZ(\s)\bigr) +
\bbP_\s\bigl(|\cZ(\s)\setminus
\cZ(\s_{T_n})|\ge 2\bigr)
+\bbP_\s^C\bigl(|\cZ(\s)\setminus
\cZ(\s_{T_n})|\ge 2\bigr)\right]\label{delta1}\\
+ 2 \sup_{\s\in \O_\L^{(\ge n)}}\sum_{x\in
\cZ(\s)}\Big|\,\bbP_\s\bigl(\cZ(\s_{T_n})=\cZ(\s)\setminus
\{x\}\bigr)-\bbP_\s^C\bigl(\cZ(\s_{T_n})=\cZ(\s)\setminus
\{x\}\bigr)\, \Big |\,. \label{delta2}
\end{gather}
The contribution in \eqref{delta1} can be bounded, using Corollary
\ref{sameforHC}, Lemma \ref{zeri} and Corollary \ref{twozeros} for $n\geq 1$ and by an easy calculation in the case $n=0$, by
$c\left(q + \left(T_n/t_n\right)^2\right) $ for some
constant $c=c(L,N)$ and therefore, when multiplied by $M_n\le t_n^+/T_n$,
vanishes as $q\downarrow 0$.

The contribution in \eqref{delta2} is instead bounded from above by
\begin{gather*}
c\sup_{\s\in \O_\L^{(\ge n)}}\Big( \g \sup_{x\in \cZ(\s)}\bbP^C_\s\bigl(\cZ(\s_{T_n})=\cZ(\s)\setminus
\{x\}\bigr)+  \sup_{\stackrel{x\in \cZ(\s)}{x
    \text{ of class } \ge n+1}}\bbP_\s\bigl(\cZ(\s_{T_n})=\cZ(\s)\setminus
\{x\}\bigr)\Bigr)
\end{gather*}
by the definition \eqref{gamma} of $\g$ (recall  that any zero $x$
of class at least $n+1$ cannot be erased during the $n$--th
coalescence process). Because of Corollary \ref{subszeros}, uniformly
in $\s\in \O_\L^{(\ge n)}$,
\[
\sup_{\stackrel{x\in \cZ(\s)}{x \text{ of class } \ge n+1}}\bbP_\s\bigl(\cZ(\s_{T_n})=\cZ(\s)\setminus
\{x\}\bigr)\le  c T_n/t_{n+1}\]
and therefore, when multiplied by $M_n$ tends to zero as $q\downarrow 0$.
Similarly, using Lemma \ref{rates bounds},
\[
\sup_{x\in
\cZ(\s)}\bbP^C_\s\bigl(\cZ(\s_{T_n})=\cZ(\s)\setminus \{x\}\bigr)
\le c  T_n/t_n \,.
\]
which, once it is multiplied by $M_n$, is bounded from above by $c t_n^+/t_n$.

In conclusion, in order to show that $\lim_{q\downarrow 0}M_n\delta =0 $ and $\lim_{q\downarrow
0}\gamma=0 $, it is enough to
show that $\lim_{q\downarrow 0} \g\, t_n^+/t_n=0$ uniformly in $\s\in
\O_\L^{(\ge n)}$.
For this purpose, given $x\in \cZ(\s)$ with domain $d_x\in \cC_n$,
we assume that the closest zero of $\s$ to the left of $x$ is also of class $n$. Call
$z$ its position. The case in which this assumption is not
verified can be treated analogously.

Then we write
\begin{gather*}
  \bbP^C_\s\bigl(\cZ(\s_{T_n})=\cZ(\s)\setminus
\{x\}\bigr)={\rm Prob}\left(\xi_x\le T_n, \ \xi_x \le \xi_z \right)
\prod_{\stackrel{y\in \cZ(\s)}{y\neq x,z \text{ is of class $n$}}}{\rm Prob}\left(\xi_y \ge T_n \right)
\end{gather*}
where $\{\xi_y\}$,
$y\in \cZ(\s)$,  are independent
exponential variables with parameter $\l_n(d_y)$ (recall the graphical
construction in Section \ref{coppietta}).

Using the definition \eqref{deflambda} of the rates $\l_n$ and
\eqref{hit2} of Lemma
\ref{prop:hitting}, for any $y\in \cZ(\s)$ of class $n$
\[{\rm Prob}\left(\xi_y \ge T_n \right)=
\bbP^{[0,d_y-1]}_{0\mathds{1}}(\tilde \tau \geq T_n)=1 + O(T_n/t_n)
\]
where $\tilde \tau$ is the hitting time of the set
$\{\s:\s(0)=0\}$.
On the other hand, for the same reasons,
\begin{gather*}
{\rm Prob}(\xi_x\le T_n)\ge  {\rm Prob}\left(\xi_x\le T_n, \ \xi_x \le
  \xi_z \right) \ge {\rm Prob}\left(\xi_x\le T_n\right) - {\rm
  Prob}\left(\xi_x\le T_n,\  \xi_z\le T_n\right)\\
 = \bbP^{[0,d_x-1]}_{0\mathds{1}}(\tilde \tau \leq T_n)(1+O(T_n/t_n))\,.
\end{gather*}
We conclude that, uniformly in $\s$,
\begin{equation}
  \label{eq:PC}
\bbP^C_\s\bigl(\cZ(\s_{T_n})=\cZ(\s)\setminus \{x\}\bigr)=
\bbP^{[0,d_x-1]}_{0\mathds{1}}(\tilde \tau \leq T_n)(1+
O(T_n/t_n))\,.
\end{equation}
Similarly, with $\t_n(y)$ the first time such that there are $n$
zeros strictly inside the domain of $y\in \cZ(\s)$, we can write
\begin{gather}
  \bbP_\s\bigl(\cZ(\s_{T_n})=\cZ(\s)\setminus \{x\}\bigr)\nonumber\\
  = \bbP_\s\left( \cZ(\s_{T_n})=\cZ(\s)\setminus \{x\},\ \t_n(y)> T_n\
    \forall y\neq x \right) \nonumber\\ + \bbP_\s\left(
    \cZ(\s_{T_n})=\cZ(\s)\setminus \{x\} ,\  \t_n(y)\le  T_n \
    \text{ for some } y\neq x\right)
\label{AAA}
\end{gather}
The last term, thanks to Corollary \ref{twozeros}, is bounded from above by $c(T_n/t_n)^2$.
Thanks to Lemma \ref{oriented} the first term in the r.h.s. factorizes as
\[
\bbP^V_\h\left(\s_{T_n}=\s_{0\mathds{1}}\right) \prod_{\stackrel{y\in
      \cZ(\s)}{y\neq x,z}}\bbP_{0\mathds{1}}^{[0,d_y-1]}(\t_n >T_n,\ \s_{T_n}=\s_{0\mathds{1}})
\]
where $V=[0,d_z+d_x-1]$ and $\h\in \O_V$ is such that $\cZ(\h)=\{0,d_z\}$.
Indeed, if $\t_n(y)> T_n \
    \forall y\neq x $, all the zeros in $\cZ(\s)$ different from
    $x$ are frozen thanks to Remark \ref{remarkzero}.
By Lemma \ref{zeri} and \eqref{hit2} in Lemma \ref{prop:hitting}
\[
\prod_{\stackrel{y\in
      \cZ(\s)}{y\neq x,z}}\bbP_{0\mathds{1}}^{[0,d_y-1]}(\t_n >T_n,\
  \s_{T_n}=\s_{0\mathds{1}})=1 + O(q) + O(T_n/t_n)= 1+O(T_n/t_n)
\]
On the other hand, conditioned on the $\s$-algebra $\cF$ of the Poisson
processes and coin tosses associated to $[d_z,d_z+d_x-1]$,
\[
\bbP^V_\h\left(\cZ(\s_{T_n})\cap [0,d_z-1]=\{0\}\tc \cF\right) = 1 +
O(q) + O(T_n/t_n)=1+O(T_n/t_n)
\]
because of \eqref{eq2} and \eqref{eq2.1} in Lemma \ref{zeri} and Remark \ref{remarkzero}.
Hence
\[
\bbP^V_\h\left(\s_{T_n}=\s_{0\mathds{1}}\right)=\bbP^{[0,d_x-1]}_{0\mathds{1}}\left(\s_{T_n}=\s_{\mathds{1}}\right)(1+O(T_n/t_n))\,.
\]
Finally
\[
\bbP^{[0,d_x-1]}_{0\mathds{1}}\left(\s_{T_n}=\s_{\mathds{1}}\right)=1-
\bbP^{[0,d_x-1]}_{0\mathds{1}}\left(\s_{T_n}\neq
  \s_{\mathds{1}}\right)\ge
1-\bbP^{[0,d_x-1]}_{0\mathds{1}}\left(\tilde \tau \ge T_n\right)+O(q)\ge c T_n/t_n
\]
because of  \eqref{same2.1}.

Going back to \eqref{AAA} and collecting the above estimates we
conclude that
\begin{equation}
  \label{P}
\bbP_\s\bigl(\cZ(\s_{T_n})=\cZ(\s)\setminus \{x\}\bigr)=
\bbP^{[0,d_x-1]}_{0\mathds{1}}\left(\s_{T_n}=\s_{\mathds{1}}\right)(1+O(T_n/t_n))
\end{equation}
If we put together \eqref{eq:PC} and \eqref{P} we get
\[
\g\le \Big|  \frac{\bbP^{[0,d_x-1]}_{0\mathds{1}} (\s_{T_n}=\s_\mathds{1})}{\bbP^{[0,d_x-1]}_{0\mathds{1}}(\tilde
\tau \leq T_n)}-1\Big| +O(T_n/t_n)\,.
\]
The contribution of the error term $O(T_n/t_n)$ to $\g\, t_n^+/t_n$ tends to zero as $q\downarrow 0$. As far as the first term is concerned we can write
\begin{gather}
1-\frac{\bbP^{[0,d_x-1]}_{0\mathds{1}}(\s_{T_n}=\s_\mathds{1})}{\bbP^{[0,d_x-1]}_{0\mathds{1}}(\tilde
  \tau \leq T_n)}\nonumber \\
= \frac{\bbP^{[0,d_x-1]}_{0\mathds{1}}(\{\tilde \tau \leq
T_n\}\cap\{\s_{T_n}\neq\s_\mathds{1}\})}{\bbP^{[0,d_x-1]}_{0\mathds{1}}(\tilde
\tau \leq T_n)} \le c\,q/\bbP^{[0,d_x-1]}_{0\mathds{1}}(\tilde \tau
\leq T_n)\le c\, q \frac{t_n}{T_n}\,.
\label{CCC}
\end{gather}
In the first inequality we used the bound (see Lemma \ref{zeri})
\[
\bbP^{[0,d_x-1]}_{0\mathds{1}}(\{\tilde \tau \leq
T_n\}\cap\{\s_{T_n}(y)=0\})\le
\begin{cases}
  \bbP^{[0,d_x-1]}_{0\mathds{1}}(\s_{T_n}(y)=0)\le cq &\text{if $y\neq 0$}\\
 \bbP^{[0,d_x-1]}_{0\mathds{1}}(\{\tilde \tau \leq T_n\}\cap\{\s_{T_n}(0)=0\})\le cq &
\end{cases}
\]
where, for the case $y=0$, the estimate follows from the strong Markov property and Lemma \ref{zeri} applied to the starting configuration $\s_{\tilde \tau}$.
In the second inequality we used \eqref{same2.1} to get  $\bbP^{[0,d_x-1]}_{0\mathds{1}}(\tilde \tau \leq T_n)\ge c T_n/t_n$.

Since $\lim_{q\downarrow 0}\frac{t_n^+}{t_n} \, q\, \frac{t_n}{T_n}=0$ we can conclude that $\lim_{q\downarrow 0} \g \frac{t_n^+}{t_n}=0$ and the proof is complete.
\end{proof}

\section{Proof of Theorems \ref{plateau} and \ref{asymptotics}}
\label{sectionproofs}
\begin{proof}[Proof of Theorem \ref{plateau}]\\
(i) Thanks to Lemma \ref{survivallemma}
\[
\varlimsup _{q\downarrow 0}\, \sup _{t \in
  [t_{n}^+,
  t_{n+1}^-]}\left|\bbP_Q(\sigma_t(0)=0)-\bbP_Q(\sigma_s(0)=0\ \forall
  s\le t)\right|=0\,.
\]
Hence it is enough to prove that
\[
\varlimsup _{q\downarrow 0}\, \sup _{t \in [t_{n}^+,
t_{n+1}^-]}\left|\bbP_Q(\sigma_t(0)=0)-\left(\frac{1}{2^{n}+1}\right)^{c_0(1+o(1))
}\right|=0
\]
where $o(1)$ is an error term going to zero as $n\to \infty$. Equation \eqref{bacio2} of Theorem \ref{bacio} tells us that
\[
\varlimsup _{q\downarrow 0}\, \sup _{t \in
[t_{n}^+,
t_{n+1}^-]}\left|\bbP_Q(\sigma_t(0)=0)-\bbP_Q^H(x_0(t)=0)\right|=0.
\]
In turn, thanks to Remark \ref{remhcp} it holds
\[
\varlimsup _{q\downarrow 0}\, \sup _{t \in
[t_{n}^+,
t_{n+1}^-]}\left|\bbP_Q^H(x_0(t)=0)-\bbP_Q^H(\sigma_0^{(n+1)}(0)=0)\right|=0.
\]
and (iv) of Theorem \ref{anguriona} says that
\[
\bbP_Q^H(\sigma_0^{(n+1)}(0)=0)=\frac{1}{\left(2^{n}+1\right)^{c_0
(1+o(1))}}\,
\]
and the sought result follows.

\noindent(ii)
The result follows immediately by using  Lemma \ref{survivallemma}.

\noindent (iii) Fix $x\in \bbZ_+$, $m<n$ and $s\in [t_m^+,t_{m+1}^-],\ t\in [t_n^+,t_{n+1}^-]$. Because of Lemma \ref{zeri}
\[
\varlimsup _{q\downarrow 0}\sup _{\stackrel{t \in [t_{n}^+,
t_{n+1}^-]}{s \in [t_{m}^+, t_{m+1}^-]}}\bbP_Q(\s_t(x)=0\tc
\s_s(x)=1)=0 \,.\] Hence
\begin{gather*}
C_Q(s,t,x)= \bbP_Q(\s_t(x)=0\cap \s_s(x)=0)-\bbP_Q(\s_t(x)=0) \bbP_Q(\s_s(x)=0)
\\
=\bbP_Q(\s_t(x)=0)\left(1 - \bbP_Q(\s_s(x)=0)\right) +\d(s,t,q)
\end{gather*}
with
\[
\varlimsup _{q\downarrow 0}\sup _{\stackrel{t \in [t_{n}^+,
t_{n+1}^-]}{s \in [t_{m}^+, t_{m+1}^-]}}\d(s,t,q)=0\,.
\]
Similarly
\[
\varlimsup _{q\downarrow 0}\sup _{t \in [t_{n}^+,
t_{n+1}^-]}\left|\bbP_Q(\s_t(x)=0)- \bbP_Q(\s_t(x)=0\,,\,
\s_0(x)=0)\right|=0
\]
and the same at time $s$. Since $\bbP_Q(\s_t(x)=0\tc \s_0(x)=0)=\bbP_Q(\s_t(0)=0)$
because of the renewal property of $Q$ the proof follows at once from part (i).
\end{proof}

\begin{proof}[Proof of Theorem \ref{asymptotics}]
The proof follows at once from Theorems \ref{anguriona}, \ref{bacio} and Remark \ref{remhcp}.
\end{proof}

\section{Extensions} \label{extensions}

In this section we present  some extensions of our results. As
already mentioned,  the technical assumption that the interval law
$\mu$ satisfies $\mu([k,\infty))>0$ for all $k \in \bbN$ can be
removed as discussed in  \cite{FMRT-Cergy}. As a consequence, in
what follows we disregard this assumption.

\subsection{The East process on $\bbZ$ with renewal stationary   initial distribution}
We say that a random subset $\s$  of $\bbZ$ is stationary if its law
$Q$ is left invariant by any  translation along a vector $x \in
\bbZ$. In addition, we say that it is renewal if the law $Q(\cdot
\tc 0 \in \s)$ equals ${\rm Ren}(\mu\tc 0)$ for some probability
measure on $\bbN$ (note that   $Q(0 \in \s)$ must be positive due to
the stationarity). For simplicity, we  write $Q= {\rm Ren}(\mu)$ and
call $\mu$ the interval law. Note that  under the Bernoulli
probability on $\{0,1\}^\bbZ$ with parameter $p$, the set of zeros
has law ${\rm Ren}(\mu\tc 0)$ with $\mu(n)=p^{n-1}(1-p)$.

It can be proved (see \cite{DV}) that $\mu$ must have finite mean.
Moreover, $Q$--a.s. the random subset $\s$ is given by infinite
points $\{x_k\}_{k\in \bbZ}$ with $\lim _{k\to \pm \infty } x_k=\pm
\infty$. In what follows, we enumerate the points in $\s$ with the
convention that $x_k<x_{k+1}$ and $x_0\leq 0 <x_1$.

Due to formula (C3) in \cite[Appendix C]{FMRT} the set $\s$ with law
$Q={\rm Ren}(\mu)$ is characterized by the following properties:
\begin{itemize}
\item[(i)] the points $x_0\leq 0 <x_1$ have law
\begin{equation}\label{puffi}
\begin{split}  Q(x_0\leq -m\,,\; x_1\geq n)& =(1/\bar\mu)  \sum _{\ell =n+m}^\infty \mu(\ell) (\ell
-n-m+1)\\
& =(1/\bar\mu) \sum _{\ell =n+m}^{\infty}\mu([\ell, \infty))
\end{split}
\end{equation}
for all  $n\geq 1, m\geq 0$, where $\bar \mu$ denotes the average of
$\mu$: $\bar \mu=\sum _{\ell=1}^\infty \mu(\ell) \ell $;
\item[(ii)]
the domain lengths $(x_{k+1}-x_k)$, $k\in \bbZ\setminus\{0\}$, are
i.i.d. random variables with law $\mu$ and are also independent from
$x_0,x_1$.
\end{itemize}
Note that \eqref{puffi} implies that  $Q(0\in \xi)= 1/\bar \mu $ and
that $|x_0|+1$ has the same law of $x_1$. Due to the above
characterization,
  under $Q={\rm  Ren} (\mu)$ the law
of $\s \cap \bbZ_+$ is given by $Q_+= {\rm Ren} ( \nu,\mu)$, where
$\nu$ is the   probability measure  on $\bbZ_+$ such that
$$\nu(n)=(1/\bar\mu)\mu([n+1, \infty))\,,\qquad n\in \bbZ_+ $$
(indeed by stationarity $\nu$ coincides with the law of $x_1-1$ since
$x_1$ is the leftmost point of $\s\cap \bbN$). Note that $\nu$ has
finite $m^{th}$--moment if and only if $\mu$ has finite
$(m+1)^{th}$-moment.

The above observation implies that Theorem \ref{plateau} can be
adapted to the stationary case following the guidelines of Remark
\ref{anyk}.

\medskip The East process starting from $Q$  must be compared
with the hierarchical coalescence process on $\bbZ$ starting from
$Q$ (the definition is a straightforward extension of the one given
for the HCP on $\bbZ_+$). In \cite{FMRT} it is proved that,
considering the HCP on $\bbZ$ starting with distribution $Q={\rm
Ren}(\mu)$,  the law at the beginning of the $n$--th epoch is simply
${\rm Ren}(\mu_n)$ with $\mu_n$ defined from $\mu$ as in Subsection
\ref{gerarchia}.

  By slightly modifications in the proof, Theorem
\ref{bacio} becomes:
\begin{Theorem}\label{baciobis}
For any $N\in \bbN$ let $\epsilon_N:=1/8N$ and choose the parameter
$\epsilon$ appearing in Definition \ref{stalling-active} and in
\eqref{Tn} equal to $\epsilon_N$.  Let   $Q=\text{Ren}(\mu)$ with
$\mu$ probability measure on $\bbN$ having  finite mean. Then for
any $k\in \bbZ$
\begin{equation*}
\lim _{q\downarrow 0}\, \sup _{t \in [0, t_N^+]} \,
 d_{TV}\bigl(\, \{(x_{-k} (t), \dots, x_k(t) ) ,\, \bbP _Q\}\,;\, \{(x^H_{-k} (t), \dots, x^H_k(t) ) ,\, \bbP^H _Q\}\bigr)= 0 \,.
\end{equation*}
Assume that the $(m+\d)^{th}$-moment of $\mu$  is finite for some
$\d>0$. Then, for each $k \in \bbZ\setminus \{0\}$ it holds
\begin{equation*}
 \lim _{q \downarrow 0} \sup_{t\in [0, t_N^+] }
\bigl| \bbE_Q \bigl( [x_{k+1} (t)-x_{k} (t)]^m
\bigr)-\bbE_Q^\text{H}  \bigl( [x^H_{k+1} (t)-x^H_{k} (t)]^m
\bigr)\bigr|=0 \,.
\end{equation*}
Assume  that the $(m+1+\d)^{th}$-moment of $\mu$  is finite for
some $\d>0$. Then the above equation \eqref{momentaccio1} is valid
also  for $k=0$ and moreover, for all $k \in \bbZ$, it holds
\begin{equation*}
\lim _{q \downarrow 0} \sup_{t\in [0, t_N^+] } \bigl| \bbE_Q \bigl(
[x_k  (t)]^m \bigr)-\bbE_Q^\text{H}  \bigl( [x^H_k (t)]^m
\bigr)\bigr|=0 \,.
\end{equation*}
\end{Theorem}
Using the above approximation result and the scaling limits
discussed in \cite{FMRT}, Theorem \ref{asymptotics} remains valid in
the stationary case by setting
$$
\bar X^{(n)}(t):=(x_{k+1}(t)-x_k(t))/(2^{n-1}+1)\quad ;\quad \bar
Y^{(n)}(t):=x_1(t)/(2^{n-1}+1),
$$
where $k$ is any integer in $\bbZ\setminus \{0\}$.  $\bar
Y^{(n)}(t)$ can also be defined as $|x_0(t)|/(2^{n-1}+1 )$.

\begin{remark}
The   above extensions to the stationary case, and their derivation,
will be discussed in more detail in \cite{FMRT-Cergy}. There we will
present other results, including the aging through hierarchical
coalescence in the East process on the half--line $\{-1,-2, \dots\}$
with frozen zero at site $0$.
\end{remark}

\subsection{The East process on $\bbZ_+$ with exchangeable initial distribution}\label{iguana}
Our main results, with suitable modifications, can be formulated
also when the initial distribution in an {\sl exchangeable} one. We
say that the law $Q$ of a random set of points
$\{x_i\}_{i=0}^\infty$ in $\bbZ_+$ containing the origin  is
exchangeable if this set  has infinite cardinality a.s.\  and  the
law of the random sequence $x_1-x_0=x_1$, $x_2-x_1$, $x_3-x_2$,... is
invariant w.r.t. finite permutations. By  De Finetti Theorem, $Q$
can be expressed as  $Q= \int _{\Upsilon} \mathfrak{p}(d\z) Q_\z$,
where $Q_\z= {\rm Ren}(\mu_\z\tc 0)$ and the parameter $\z$ varies
on a probability space $(\Upsilon,\mathfrak p)$ \cite[Appendix
D]{FMRT}.

Considering the East process on $\bbZ_+$ with initial distribution
$Q$, suppose that for $\mathfrak{p}$--a.a. $\z\in \Upsilon $ the law
$\mu_\z$ satisfies condition (a) or (b) in Theorem \ref{plateau},
set $c_0(\z)=1$ and $c_0(\z)=\a$ respectively. Then Theorem
\ref{plateau} remains valid introducing in the asymptotic values the
average $\int _{\Upsilon} \mathfrak{p}(d\z)$ and replacing $c_0$ with
$c_0(\z)$ and $\rho_x $ with $Q_\z (\s(x)=0)$. By similar
modifications, also Theorem \ref{asymptotics} remains valid for an
exchangeable $Q$. Clearly, the average over $\mathfrak p(d\z)$ may lead to new asymptotic behaviors.
Finally, Theorem \ref{bacio} still holds provided that the
$(m+\d)^{th}$-moment w.r.t $Q$ of $(x_{k+1}-x_k)$ (which is
$k$--independent by exchangeability) is finite.

\subsection*{Acknowledgements}
We  thank   the Laboratoire de Probabilit\'{e}s et Mod\`{e}les
Al\'{e}atoires, the University Paris VII and the Department of
Mathematics of the University of Roma Tre for the  support and the
kind hospitality.
C. Toninelli acknowledges the partial support of the
French Ministry of Education
through the ANR BLAN07-2184264 grant.

\bibliographystyle{amsplain}
\bibliography{East}

\providecommand{\bysame}{\leavevmode\hbox to3em{\hrulefill}\thinspace}
\providecommand{\MR}{\relax\ifhmode\unskip\space\fi MR }
\providecommand{\MRhref}[2]{%
  \href{http://www.ams.org/mathscinet-getitem?mr=#1}{#2}
}
\providecommand{\href}[2]{#2}
\begin{thebibliography}{10}

\bibitem{Aldous}
D.~Aldous and P.~Diaconis, \emph{The asymmetric one-dimensional constrained
  ising model: rigorous results}, J. Stat. Phys. \textbf{107} (2002), no.~5-6,
  945--975.

\bibitem{A}
A.~Asselah and P.~Dai Pra, \emph{Quasi-stationary measures for conservative
  dynamics in the infinite lattice}, Ann. Probab. \textbf{29} (2001), no.~4,
  1733--1754.

\bibitem{D3}
A.J. Bray, B.~Derrida, and C.~Godr\`{e}che, \emph{Non--trivial algebraic decay
  in a soluble model of coarsening}, Europhys. Lett. \textbf{27} (1994), no.~3,
  175--180.

\bibitem{CMRT}
N.~Cancrini, F.~Martinelli, C.~Roberto, and C.~Toninelli, \emph{Kinetically
  constrained spin models}, Probability Theory and Related Fields \textbf{140}
  (2008), no.~3-4, 459--504.

\bibitem{CMST}
N.~Cancrini, F.~Martinelli, R.~Schonmann, and C.~Toninelli, \emph{Facilitated
  oriented spin models: some non equilibrium results}, J. Stat. Phys.
  \textbf{138} (2010), no.~6, 1109--1123.

\bibitem{CDG}
F.~Chung, P.~Diaconis, and R.~Graham, \emph{Combinatorics for the east model},
  Adv. in Appl. Math. \textbf{27} (2001), no.~1, 192--206.

\bibitem{CorberiCugliandolo}
F.~Corberi and L.~F. Cugliandolo, \emph{Out-of-equilibrium dynamics of the
  spiral model}, J. Stat. Mech. (2009), P09015.

\bibitem{Crisanti}
A.~Crisanti, F.~Ritort, A.~Rocco, and M.~Sellitto, \emph{Inherent structures
  and non-equilibrium dynamics of 1d constrained kinetic models: a comparison
  study}, J. Chem. Phys. \textbf{113} (2000), 10615--10647.

\bibitem{DV}
D.J. Daley and D.~Vere-Jones, \emph{An introduction to the theory of point
  processes.}, second ed., Probability and its Applications (New York),
  Springer, New York, 2008, General theory and structure.

\bibitem{D0}
B.~Derrida, \emph{Coarsening phenomena in one dimension}, Complex systems and
  binary networks ({G}uanajuato, 1995), Lecture Notes in Physics, vol. 461,
  Springer, Berlin, 1995, pp.~164--182.

\bibitem{D1}
B.~Derrida, C.~Godr\`{e}che, and I.~Yekutieli, \emph{Stable distributions of
  growing and coalescing droplets}, Europhys. Lett. \textbf{12} (1990), no.~5,
  385--390.

\bibitem{D2}
\bysame, \emph{Scale-invariant regimes in one-dimensional models of growing and
  coalescing droplets}, Physical Review A \textbf{44} (1991), no.~10,
  6241--6251.

\bibitem{Durrett}
R.~Durrett, \emph{Lecture notes on particle systems and percolation}, Lecture
  Notes in Mathematics (1995), no.~1608.

\bibitem{FMRT}
A.~Faggionato, F.~Martinelli, C.~Roberto, and C.~Toninelli, \emph{Universality
  in one dimensional hierarchical coalescence processes}, preprint (2010).

\bibitem{FMRT-Cergy}
A.~Faggionato, F.~Martinelli, and C.~Roberto, C.~Toninelli, in preparation.

\bibitem{FA1}
G.H. Fredrickson and H.C. Andersen, \emph{Kinetic ising model of the glass
  transition}, Phys.~Rev.~Lett. \textbf{53} (1984), 1244--1247.

\bibitem{FA2}
\bysame, \emph{Facilitated kinetic ising models and the glass transition}, J.
  Chem. Phys. \textbf{83} (1985), 5822--5831.

\bibitem{GarrahanNewman}
J.~P. Garrahan and M.~E.~J. Newman, \emph{Inherent structures and
  non-equilibrium dynamics of 1d constrained kinetic models: a comparison
  study}, Phys. Rev. E \textbf{62} (2000), 7670--7680.

\bibitem{JACKLE}
J.~J\"{a}ckle and S.~Eisinger, \emph{A hierarchically constrained kinetic ising
  model}, Z. Phys. B: Condens. Matter \textbf{84} (1991), no.~1, 115--124.

\bibitem{GarrahanSollichToninelli}
C.Toninelli J.P.~Garrahan, P.Sollich, \emph{Kinetically constrained models}, to
  appear in "Dynamical heterogeneities in glasses, colloids, and granular
  media", Oxford Univ.Press, Eds.: L. Berthier, G. Biroli, J-P Bouchaud, L.
  Cipelletti and W. van Saarloos. Preprint arXiv:1009.6113.

\bibitem{Leonard}
S.~Leonard, P.~Mayer, P.~Sollich, L.~Berthier, and J.P. Garrahan,
  \emph{Non-equilibrium dynamics of spin facilitated glass models}, J. Stat.
  Mech. (2007), P07017.

\bibitem{Liggett1}
T.M. Liggett, \emph{Interacting particle systems}, Grundlehren der
  Mathematischen Wissenschaften [Fundamental Principles of Mathematical
  Sciences], vol. 276, Springer-Verlag, New York, 1985.

\bibitem{Liggett2}
\bysame, \emph{Stochastic interacting systems: contact, voter and exclusion
  processes}, Grundlehren der Mathematischen Wissenschaften [Fundamental
  Principles of Mathematical Sciences], vol. 324, Springer-Verlag, Berlin,
  1999.

\bibitem{MOS}
Fabio Martinelli, Enzo Olivieri, and Elisabetta Scoppola, \emph{Small random
  perturbations of finite- and infinite-dimensional dynamical systems:
  unpredictability of exit times}, J. Statist. Phys. \textbf{55} (1989),
  no.~3-4, 477--504.

\bibitem{Olivieri-Vares}
E.~Olivieri and M.E. Vares, \emph{Large deviations and metastability},
  Encyclopedia of Mathematics and its Applications, vol. 100, Cambridge
  University Press, Cambridge, 2005.

\bibitem{Ritort}
F.~Ritort and P.~Sollich, \emph{Glassy dynamics of kinetically constrained
  models}, Advances in Physics \textbf{52} (2003), no.~4, 219--342.

\bibitem{SE2}
P.~Sollich and M.R. Evans, \emph{Glassy time-scale divergence and anomalous
  coarsening in a kinetically constrained spin chain}, Phys. Rev. Lett
  \textbf{83} (1999), 3238--3241.

\bibitem{SE1}
\bysame, \emph{Glassy dynamics in the asymmetrically constrained kinetic ising
  chain}, Phys. Rev. E (2003), 031504.

\end{thebibliography}

\end{document}